\documentclass[11pt,onecolumn]{IEEEtran}
\usepackage[left=2cm,top=2cm, bottom=2in, right=2cm]{geometry}
\date{}
\usepackage{enumitem}
\usepackage{amsfonts}
\usepackage{amsmath, amsthm, amssymb, dsfont}
\usepackage{graphicx}
\usepackage{comment, cite}
\usepackage{psfrag}
\usepackage{bbm}
\usepackage{csquotes}

\usepackage[dvipsnames]{xcolor}
\usepackage{tikz}

\usetikzlibrary{fit}					
\usetikzlibrary{backgrounds}

\newtheorem{theorem}{Theorem}
\newtheorem{lemma}{Lemma}

\newtheorem{proposition}{Proposition}

\newtheorem{corollary}{Corollary}
\newtheorem{definition}{Definition}
\newtheorem{example}{Example}

\newtheorem{remark}{Remark}

\definecolor{darkblack}{rgb}{0, .07, .5}
\definecolor{darkred}{rgb}{0.5,0,0}
\definecolor{mahogany}{rgb}{0.65, 0., 0.5}

\newcommand{\Var}{\text{\rm{Var}}}

\newcommand{\E}{\mathbb{E}}

\let\emptyset\varnothing

\newcommand{\RN}[1]{%
	\textup{\uppercase\expandafter{\romannumeral#1}}%
}

\title{Coding for Positive Rate in the Source Model Key Agreement Problem}
\author{Amin Gohari, Onur~G\"unl\"u, and
        Gerhard~Kramer
}

\begin{document}
	\maketitle

\begin{abstract}
A two-party key agreement problem with public discussion, known as the source model problem, is considered. By relating key agreement to hypothesis testing, a new coding scheme is developed that yields a sufficient condition to achieve a positive secret-key (SK) rate  in terms of R{\'e}nyi divergence. The merits of this coding scheme are illustrated by applying it to an erasure model for Eve's side information, and by deriving an upper bound on Eve's erasure probabilities for which the SK capacity is zero. This bound strictly improves on the best known single-letter lower bound on the SK capacity. Moreover, the bound is tight when Alice's or Bob's source is binary, which extends a previous result for a doubly symmetric binary source. The results motivate a new measure for the correlation between two random variables, which is of independent interest.
\end{abstract}

\section{Introduction}
The source model problem for key agreement considers two legitimate parties, Alice and Bob, and an eavesdropper, Eve \cite{AhlswedeCsiszar,Maurer93}. Alice, Bob, and Eve, respectively, observe $n$ independent and identically distributed (i.i.d.) realizations of the random variables $X$, $Y$, and $Z$ that are distributed according to the probability mass function (pmf) $p_{XYZ}$ with values $p_{XYZ}(x,y,z)$ for $x\in\mathcal{X}$, $y\in\mathcal{Y}$, $z\in\mathcal{Z}$, where $\mathcal{X}$, $\mathcal{Y}$, and $\mathcal{Z}$ are finite sets. The pmf $p_{XYZ}$ is called the \emph{source} pmf, or simply the source. Alice and Bob engage in an authenticated and public discussion to agree on a key secret from Eve as follows (see \cite[Section 22.3]{ElGamalKim} for a  review of the problem). Alice first creates a public message $F_1=f_1$ using some distribution $p_{F_1|X^n}(\cdot |x^n)$ and sends it to Bob, where $x^n$ is the sequence $(x_1,\ldots,x_n)$. Bob generates a public message $F_2=f_2$ using some $p_{F_2|Y^nF_1}(\cdot |y^n,f_1)$ and sends it to Alice, then Alice generates $F_3=f_3$ according to some $p_{F_3|X^nF_{1:2}}(\cdot |x^n,f_{1:2})$, etc. After $k$ rounds of communication, Alice creates a key $K_A$ according to some $p_{K_A|X^nF_{1:k}}(\cdot |x^n,f_{1:k})$ and Bob creates a key $K_B$ according to some $p_{K_B|Y^nF_{1:k}}(\cdot |y^n,f_{1:k})$. For an $(n,\delta)$ code, we require the keys to be equal with high probability
\begin{align}
	\mathbb{P}[K_A=K_B]\geq 1-\delta \label{eq:errorprobconst}
\end{align}
where $\delta$ is small. We require the keys to have large entropy, a condition that we express as
\begin{align}
	\frac 1n H(K_A)\geq \frac 1n \log|\mathcal{K}| - \delta
\end{align}
where $H(\cdot)$ is the entropy function, $\mathcal{K}$ is the alphabet of $K_A$ and $K_B$, and $|\mathcal{K}|$ is the cardinality of $\mathcal{K}$. We also require the keys to be  almost independent of Eve's information, i.e., we require
\begin{align}
	\frac 1n I(K_A;Z^nF_{1:k})\leq \delta\label{leakjbsdfkb}
\end{align}
where $I(\cdot;\cdot)$ is the mutual information. 

The key rate is $R=\frac{1}{n}H(K_A)$. A key rate $R$ is said to be achievable if, given any $\delta>0$, there is an $(n,\delta)$ code satisfying (\ref{eq:errorprobconst})-(\ref{leakjbsdfkb}). The supremum of all achievable key rates is called the source model secret-key (SK) capacity and denoted by $S(X;Y\|Z)$. Extensions to multiple parties and continuous random variables can be found in \cite{CsNar,Chanomni,Nitinawarat,Tyagi,Chan2,chan2018optimality}.

In this paper we are interested in the feasibility of key generation at a positive rate, i.e., $S(X;Y\|Z)>0$. Note that if key generation is feasible, it should be feasible also when Alice and Bob do not have access to private randomness. Thus, we may assume that the interactive communication $\mathbf{F}$ satisfies 
\begin{align}
&H(F_1|X^n)=H(F_2|F_1Y^n)=H(F_3|F_{1:2}X^n)=\cdots=0.\label{eqn-=deertg}
\end{align} 
The idea is that Alice and Bob can use an initial part of their observations to distill private randomness and the remaining parts to extract a key.

\subsection{Overview of the Main Results}

We provide a condition for the feasibility of SK generation at a positive rate for a general pmf $p_{XYZ}$ in terms of R{\'e}nyi divergence. The condition is based on a new approach for achieving a  positive SK rate that extends ideas in  \cite{Maurer93, OrlitskyWigderson, MaurerWolf99}. The extension is motivated by using hypothesis testing for SK agreement. Hypothesis testing was previously used for infeasibility results in network information theory, e.g., meta converses \cite{polyanskiy2010channel, tyagi2015converses}. 

Our coding scheme falls in the general class of advantage distillation protocols (e.g. see \cite[Section V]{Maurer93}). Roughly speaking, advantage distillation protocols provide an information-theoretic advantage to Alice and Bob over Eve as follows: Alice and Bob select a subset of realizations of their random variables for which they have an advantage over Eve. Utilizing public discussion, Alice and Bob agree on these realizations. More information about advantage distillation protocols and their variants can be found in   \cite{MaurerParityCheckAdvantage,MaurerParityProof, muramatsu2006secret}.

On the other hand,  the particular code can be understood in terms of a hypothesis testing problem that Bob solves to find an estimate of Alice's secret key. A key feature of the problem is a ``swapping idea" where the positions of blocks of observations are swapped under different hypotheses. If Bob's error exponent in solving the problem is better than Eve's error exponent, then Alice and Bob can use privacy amplification to obtain a shared secret key.

We illustrate the merits of this construction for the class of joint pmfs $p_{XYZ}=p_{XY}\,p_{Z|XY}$ where $p_{Z|XY}$ is an erasure channel, i.e., $\mathcal{Z}=\{\mathtt e\}\cup (\mathcal{X}\times\mathcal{Y})$ where $\mathtt e$ is the erasure symbol and $Z=XY$ with probability $1-\epsilon$ and $Z=\mathtt{e}$ with probability $\epsilon$. Such sources are called \emph{erasure sources} with erasure probability $\epsilon$. If $Z=XY$ then we have $S(X;Y\|Z)=0$, and if $Z=\mathtt e$ then we have $S(X;Y\|Z)=I(X;Y)$. We provide a sufficient condition on $\epsilon$ such that $S(X;Y\|Z)>0$. We also prove the necessity of this condition when $X$ or $Y$ is binary.

For example, consider the doubly symmetric binary-erasure (DSBE) source, where $(X,Y)$ is a doubly symmetric binary source (DSBS) with parameter $p$, i.e., we have $p_{XY}(0,0)=p_{XY}(1,1)=(1-p)/2$ and $p_{XY}(0,1)=p_{XY}(1,0)=p/2$ (see \cite[Scenario 2 and Fig. 6(a)]{MaurerWolf99}). We present an example where our code outperforms the code given in \cite{ourpaper} that uses random binning and multiple auxiliary random variables to capture multiple rounds of communication. Our achievability result shows that the SK capacity vanishes if and only if
\begin{align}
	\epsilon\leq\frac{\min\{p,1-p\}}{\max\{p,1-p\}}
\end{align}
which recovers the result in \cite[Theorem 14]{MaurerWolf99} that uses a repetition code protocol from \cite{Maurer93} for advantage distillation. However, we show more, namely that our code (and the repetition code protocol of \cite{Maurer93}) outperforms the best known single-letter lower bound given in \cite[Theorem 7]{ourpaper}. We show that this bound vanishes if and only if
\begin{align}
	\epsilon\leq 4p(1-p). \label{eq:fpomp}
\end{align}
Since $4p(1-p)>{\min\{p,1-p\}}/{\max\{p,1-p\}}$ when $p\neq 1/2$, the bound \eqref{eq:fpomp} is loose for a DSBE source for any $p\neq 1/2$.

This paper is organized as follows. In Section~\ref{lit-review}, we describe notation and give definitions. We also review the best known bounds for the SK agreement problem, and a characterization of when the SK capacity is positive. Our main results are given in Section \ref{main-result} and proved in Section~\ref{sec:Proofs}. 

\section{Preliminaries}\label{lit-review}
\subsection{Notation and Definitions}
Random variables are written as uppercase letters while their realizations  are written as lowercase letters. Calligraphic letters denote sets. We write $X\rightarrow Y\rightarrow Z$ if $X$ and $Z$ are statistically independent given $Y$, i.e., $p_{XZ|Y}=p_{X|Y}\,p_{Z|Y}$, and we say that $X\rightarrow Y\rightarrow Z$ forms a Markov chain. We write $U_{1:i-1}$ for $(U_1, U_2, \ldots, U_{i-1})$. The function $h(p)=-p\log p -(1-p)\log (1-p)$ is the binary entropy function. We use $\mathbf{F}$ to denote the public discussion $(F_1, F_2, \ldots, F_k)$.

\begin{definition}\label{def:renyidivv}
The R{\'e}nyi divergence of order $\alpha$ between two pmfs $p_{X}$ and $q_X$ is 
\begin{align}
&D_\alpha (p_X\| q_X) = \frac{1}{\alpha-1}\log\left(\sum_{x} p_X(x)^\alpha
q_X(x)^{1-\alpha}\right).
\end{align}
The Chernoff information between two pmfs $p_{X}$ and $q_X$ is
\begin{align}C(p_X\|q_X)=\max_{\alpha\in[0,1]}(1-\alpha)D_\alpha (p_X\| q_X).\label{eqnCI}\end{align}
\end{definition}

Observe that
\begin{align}
C(p_X\|q_X)&=-\log\left( \min_{\alpha\in[0,1]}\sum_{x} p_X(x)^\alpha
q_X(x)^{1-\alpha}\right)\nonumber
\\&\leq -\log\left(\sum_{x} \min(p_X(x), 
q_X(x))\right)\nonumber
\\&= -\log\left(1-\|p_X-q_X\|_{TV}\right)\label{TV-Chernoff}
\end{align}
where $\|p_X-q_X\|_{TV}=\frac12\sum_{x} |p_X(x)-
q_X(x)|$ is the total variation (TV) distance between $p_X$ and $q_X$.

The R{\'e}nyi divergence of order $\alpha = 1/2$ is
\begin{align}
&D_{\frac 12} (p_X\| q_X) = -2\log\left(\sum_{x} p_X(x)^{\frac 12}
q_X(x)^{\frac 12}\right)
\end{align}
which is equivalent (up to a transformation) to other distance measures such as Hellinger distance, Bhattacharyya distance, and fidelity. For instance, the Bhattacharyya distance is $(1/2)D_{\frac 12}(p_X\|q_X)$. From \eqref{eqnCI}, it is clear that
$C(p_X\|q_X)\geq (1/2)D_{\frac 12}(p_X\|q_X)$ and the following example shows that this inequality is tight in some cases.
\begin{example}\label{RnyiHalf}
Given $p_X$ and $q_X$, consider the distributions $\nu$ and $\omega$ with
\begin{align*}
  \nu(x_1,x_2)        & = p_X(x_1)q_X(x_2) \\
  \omega(x_1,x_2) & = q_X(x_1)p_X(x_2).
\end{align*}
For this special case we have
\begin{align}C(\nu\|\omega)&=
-\log\Bigg( \min_{\alpha\in[0,1]}\sum_{x_1,x_2}\Big( p_X(x_1)^\alpha q_X(x_2)^\alpha   q_X(x_1)^{1-\alpha}p_X(x_2)^{1-\alpha}\Big)\Bigg)\\
&=
-\log\Big( \min_{\alpha\in[0,1]}\big(\sum_{x_1} p_X(x_1)^\alpha q_X(x_1)^{1-\alpha}\big)    \big(\sum_{x_2} p_X(x_2)^{1-\alpha}q_X(x_2)^\alpha\big) \Big)
\\
&=\max_{\alpha\in[0,1]}(1-\alpha)D_\alpha (p_X\| q_X)+\alpha D_{1-\alpha} (p_X\| q_X)\nonumber
\\&=D_{\frac12}(p_X\|q_X) \label{eqn-last-stepDC}
\\&=\frac{1}{2} D_{\frac12}(\nu\|\omega) \nonumber
\end{align}
where \eqref{eqn-last-stepDC} holds because $(1-\alpha)D_{\alpha}$ is concave in $\alpha$ \cite[Corollary 2]{van2014renyi}, which implies that $(1-\alpha)D_{\alpha}+\alpha D_{1-\alpha}$ is concave and symmetric in $\alpha$, so it achieves its maximum at $\alpha=\frac 12$.
\end{example}

\begin{definition}\cite[p. 350]{CsiszarKornerbookFirstEdition} \cite[p. 929]{AG76}
Given a joint pmf $p_{XY}$, the strong data processing constant is 
\begin{align}
   s^*(X;Y)=\sup_{p_{U|X}}\frac{I(U;Y)}{I(U;X)}
\end{align}
where $U$ is an auxiliary random variable such that $U\rightarrow X\rightarrow Y$ forms a Markov chain. It suffices to consider $|\mathcal{U}|\leq |\mathcal{X}|+2$, where $\mathcal{U}$ and $\mathcal{X}$ are the respective alphabets of $U$ and $X$. 
\end{definition}
\begin{definition}\cite[p. 928]{AG76}
The maximal correlation coefficient $\rho_m(p_{XY})$ is defined as the maximum of Pearson's correlation coefficients over all non-constant functions $f(\cdot)$ and $g(\cdot)$, respectively, of $X$ and $Y$, i.e.,
\begin{align}\label{eq:max-correlatoin-31}
&\rho_m(p_{XY})\!=\! \max_{f(\cdot), g(\cdot)}\frac{\E\big[(f(X)\!-\!\E[f(X)])(g(Y)\!-\!\E[g(Y)])\big]}{\sqrt{\Var[f(X)]\Var[g(Y)]}}
\end{align}
where $\E[\cdot]$ and $\Var[\cdot]$ denote the expectation and variance operators, respectively.
\end{definition}
\begin{definition}
Given a channel $p_{Y|X}$, define
\begin{align}
\eta(p_{Y|X}) & =\max_{p_X} s^*(X;Y) \overset{(a)}{=} \max_{p_X} \rho^2_m(p_{X}\, p_{Y|X})\label{eqnAG687rr657}
\end{align}
where (a) follows from \cite[Theorem 8]{AG76}.
\end{definition}

\begin{definition}\cite{MaIshwar2013, Verdu} \label{defpreceq}
	Given two pmfs $p_{XY}$ and $q_{XY}$ on the alphabets $\mathcal X \times \mathcal Y$, the relation $q_{XY}\preceq p_{XY}$ represents existence of some functions $a:\mathcal{X}\rightarrow\mathbb{R}_{+}\cup\{0\}$ and $b:\mathcal{Y}\rightarrow\mathbb{R}_{+}\cup\{0\}$ such that for all $(x,y)\in\mathcal{X}\times\mathcal{Y}$ we have
	\begin{align}
		q_{XY}(x,y)=a(x)b(y)p_{XY}(x,y).\label{eq:defforPart1proof}
	\end{align}
\end{definition} 
\begin{example} 
	\normalfont
	Let $\mathcal{X}_1=\{x\in\mathcal{X}: p_X(x)>0\}$ and $\mathcal{X}_2=\{x\in\mathcal{X}: q_X(x)>0\}$. If $\mathcal{X}_2\subseteq \mathcal{X}_1$ and $q_{Y|X}=p_{Y|X}$, then  $q_{XY}\preceq p_{XY}$. To prove this result, choose $a(x)=q_X(x)/p_X(x)$ for $x\in\mathcal{X}_1$ and $a(x)=0$ for $x\notin\mathcal{X}_1$. Also, let $b(y)=1$ for all $y\in\mathcal{Y}$. 
\end{example}

\begin{example}
Given $q_X$, $q_Y$, and $p_{XY}$, consider the minimization problem
\begin{align}
   \kappa=\min_{r_{XY}:\: r_X=q_X, r_Y=q_Y}D(r_{XY} \| p_{XY}).
\end{align}
This is a convex optimization problem over a convex domain.  
By Lagrange multipliers, the solution to the minimization problem has the form 
\begin{align}
   r_{XY}(x,y)=a(x)b(y)p_{XY}(x,y)
\end{align}
for some $a(x)$ and $b(y)$. Thus, the optimizer $r_{XY}$ satisfies $r_{XY} \preceq p_{XY}$.
\end{example}

The term $\kappa$ appears in the literature in the context of hypothesis testing 
\cite{polyanskiy2012hypothesis} \cite[Theorems 5 and 8]{han1987hypothesis} and distributed detection \cite[Theorem 2]{shalaby1993note}. We are interested in this quantity for the following reason: let $\mathcal{T}^{(n)}_{q_X}$ be the set of sequences of length $n$ and type $q_X$, then we have the following lemma.
\begin{lemma}\label{LemmaMT} Let $(X^n, Y^n, Z^n)$ be i.i.d. according to $p_{XYZ}$. Take three arbitrary types $q_X$, $q_Y$, and $q_Z$ on $\mathcal{X}, \mathcal{Y}, \mathcal{Z}$, respectively.
\begin{itemize}
\item We have 
\begin{align}
 &\lim_{n\rightarrow\infty}\frac1n\log\mathbb{P}\left[ X^n\in \mathcal{T}^{(n)}_{q_X}, Y^n\in \mathcal{T}^{(n)}_{q_Y} \right]
=-\min_{r_{XY}:\: r_X=q_X, r_Y=q_Y}D(r_{XY} \| p_{XY} ).
\end{align}
\item We have
\begin{align}
     & \lim_{n\rightarrow\infty}
       \frac1n\log\mathbb{P}\left[ X^n\in \mathcal{T}^{(n)}_{q_X}, Y^n\in \mathcal{T}^{(n)}_{q_Y}, Z^n\in \mathcal{T}^{(n)}_{q_Z} \right]
      =-\min D(r_{XYZ} \| p_{XYZ} )
\end{align}
where the minimum is over $r_{XYZ}$ satisfying $r_X=q_X$, $r_Y=q_Y$, and $r_Z=q_Z$.
\item For any sequence $z^n$ of type $q_Z$, the probability 
\begin{align}
   \mathbb{P}\left[ Z^n=z^n,X^n\in \mathcal{T}^{(n)}_{q_X}, Y^n\in \mathcal{T}^{(n)}_{q_Y} \right]
\end{align}
depends only on $q_Z$ and is equal to
\begin{align}
   \frac{\mathbb{P}\left[ Z^n\in \mathcal{T}^{(n)}_{q_Z}, X^n\in \mathcal{T}^{(n)}_{q_X}, Y^n\in \mathcal{T}^{(n)}_{q_Y}\right]}{\left|\mathcal{T}^{(n)}_{q_Z}\right|}.
\end{align}
\end{itemize}
\end{lemma}
This lemma follows from the method of types.

\subsection{SK Capacity Lower Bound}
The authors of \cite{AhlswedeCsiszar} compute the one-way SK capacity from $X$ to $Y$ as
\begin{align}
&S_{\text{ow}}(X;Y\|Z)=\max_{UV\rightarrow X\rightarrow YZ}I(U;Y|V)-I(U;Z|V).
\end{align}
It suffices to consider $|\mathcal{V}|\leq |\mathcal{X}|$ and $|\mathcal{U}|\leq |\mathcal{X}|$ where $\mathcal{V}$, $\mathcal{U}$, $\mathcal{X}$ are the alphabets of $V,U,X$, respectively \cite[p. 561]{ElGamalKim}. Clearly, $S_{\text{ow}}(X;Y\|Z)$ is a lower bound on the SK capacity $S(X;Y\|Z)$.

The best known single-letter lower bound on $S(X;Y\|Z)$ uses interactive communication \cite[Theorem 7]{ourpaper}. Given random variables $U_1, U_2, \cdots, U_k$ satisfying the Markov chain conditions
\begin{align}
	&U_i\rightarrow XU_{1:i-1}\rightarrow YZ \text{  for odd $i$}\label{eq:LowerboundoddiMarkov}\\
	&U_i\rightarrow YU_{1:i-1}\rightarrow XZ\text{  for even $i$}\label{eq:LowerboundeveniMarkov}
\end{align}
and for any integer $\zeta$ such that $1\leq \zeta\leq k$, we have $S(X;Y\|Z)\geq L(X;Y\|Z)$ where
\begin{align}
&L(X;Y\|Z)=\Bigg[\sum_{\substack{i\geq \zeta\\\text{ odd $i$}}} I(U_i; Y|U_{1:i-1})-I(U_i;Z|U_{1:i-1})\Bigg]+\!\Bigg[\sum_{\substack{i\geq \zeta\\\text{ even $i$}}}\! I(U_i; X|U_{1:i-1})\!-\!I(U_i;Z|U_{1:i-1})\Bigg].\label{eq:lowerbound}
\end{align}
Using standard reduction techniques, we can restrict the cardinality $|\mathcal{U}_i|$ of $U_i$ to 
\begin{align}
	|\mathcal{U}_i|\leq \begin{cases}
	 |\mathcal{X}|\prod\limits_{l=1}^{i-1}|\mathcal{U}_l|& \text{for } i \text{ odd}, \\
	 |\mathcal{Y}|\prod\limits_{l=1}^{i-1}|\mathcal{U}_l|& \text{for } i \text{ even}.
	\end{cases}
\end{align}

Let $\bar{L}(X;Y\|Z)$ be the best possible lower bound obtained from (\ref{eq:lowerbound}). This bound  is difficult to evaluate since $\zeta$ and $k$ are arbitrary and the cardinality bounds on the sizes of $U_1,U_2,\ldots, U_k$ grow exponentially. However, one can simplify the calculations by using the ideas from \cite{MaIshwar2013, nair2013upper, Verdu}, where auxiliary random variables are represented by upper concave envelopes.

\subsection{SK Capacity Upper Bounds}
An early upper bound on $S(X;Y\|Z)$ is $\min\{I(X;Y), I(X;Y|Z)\}$ \cite{Maurer93}. This was later improved with the \emph{intrinsic mutual information} upper bound \cite[pp. 1126, Remark 2]{AhlswedeCsiszar},\cite{MaurerWolf99}:
\begin{align}
S(X;Y\|Z) & \le B_0(X;Y\|Z) \triangleq \min_{p_{J|Z}}I(X;Y|J).\label{eq:B0upperbound}
\end{align}
The idea is that if we make Eve weaker by passing $Z$ through a channel $p_{J|Z}$, then the SK capacity cannot decrease. Thus, we have
\begin{align}
   S(X;Y\|Z)\leq S(X;Y\|J)\leq I(X;Y|J)
\end{align}
where $XY\rightarrow Z\rightarrow J$ forms a Markov chain. Considering $J=\emptyset$ and $J=Z$, we recover the earlier upper bound $\min\{I(X;Y), I(X;Y|Z)\}$ on $S(X;Y\|Z)$. In evaluating $B_0(X;Y\|Z)$, it suffices to consider $|\mathcal{J}| \le |\mathcal{Z}|$ \cite{christandl2003property}.

Other upper bounds are given in \cite{RennerWolf, Watanabe, ourpaper, keykhosravi}. The best known upper bound is \cite{ourpaper} 
\begin{align}
&B_1(X;Y\|Z) =\!\inf_{p_{J|XYZ}} \Bigg[I(X;Y|J) + \max_{UV\rightarrow XY\rightarrow ZJ}I(U;J|V)\!-\!I(U;Z|V)\Bigg].\label{eq:B1upperbound}
\end{align}
See \cite{NewJournalMyself} for a discussion on the computability of this bound. The interpretation of $B_1(X;Y\|Z)$ given in \cite{ourpaper} is based on splitting the secret key into two parts so that only one part is independent of $J$. We  give a new interpretation of $B_1(X;Y\|Z)$ in Appendix \ref{appendixA}.

\subsection{Conditions for a Positive SK Capacity}
In an early work, Maurer gives an example where multiple rounds of communication are necessary to achieve a positive SK capacity  \cite[Section V]{Maurer93}. Orlitsky and Wigderson show in \cite{OrlitskyWigderson} that if the SK capacity is positive, only two rounds of communication suffice to realize a positive key rate. In fact, they give a necessary and sufficient condition for the SK capacity to be positive. We begin by reviewing their results. 

Consider some natural number $n$, and some sets $\mathcal{A}\subset \mathcal{X}^n$ and $\mathcal{B}\subset \mathcal{Y}^n$. Denote the pmf of $(X^n, Y^n, Z^n)$ conditioned on the events $X^n\in \mathcal{A}$ and $Y^n\in \mathcal{B}$ by $p_r(\cdot)$ so that $p_r(x^n, y^n, z^n)=0$ if $(x^n, y^n)\notin \mathcal{A}\times \mathcal{B}$; otherwise, we have
\begin{align}p_r(x^n, y^n, z^n)=\frac{p_{X^nY^nZ^n}(x^n, y^n,z^n)}{\mathbb{P}[X^n\in \mathcal{A}, Y^n\in \mathcal{B}]}\label{eqn:AAAd13}\end{align}
where $p_{X^nY^nZ^n}(x^n, y^n,z^n)=\prod_{i=1}^n p_{XYZ}(x_i, y_i, z_i)$ is a product distribution, whereas $p_r(x^n, y^n, z^n)$ is not necessarily a product distribution.

\begin{definition}[Orlitsky-Wigderson \cite{OrlitskyWigderson}]\label{def:advantage}
Given sets $\mathcal{A}\subset \mathcal{X}^n$ and $\mathcal{B}\subset \mathcal{Y}^n$, the legitimate users have a simple entropic advantage over the eavesdropper in $\mathcal{A}\times \mathcal{B}$ if for some (binary) function $f(X^n)$ we have
\begin{align}I_{p_r}(f(X^n); Y^n)> I_{p_r}(f(X^n); Z^n)\end{align}
where the mutual information expressions are calculated according to
$p_r(x^n, y^n, z^n)$.
\end{definition}

\begin{theorem} [Orlitsky-Wigderson \cite{OrlitskyWigderson}] \label{OrlitskyWigdersonThm}
The following three claims are equivalent.
\begin{enumerate}
\item $S(X;Y\|Z)>0$.
\item There exists some natural number $n$, and sets $\mathcal{A}\subset \mathcal{X}^n$ and $\mathcal{B}\subset \mathcal{Y}^n$, such that the legitimate users have a simple entropic advantage over the eavesdropper in $\mathcal{A}\times \mathcal{B}$. 
\item  A positive rate is achievable with only two rounds of communication.
\end{enumerate}
\end{theorem}
\begin{proof} The work \cite{OrlitskyWigderson} does not include proofs. We therefore give a sketch of a proof based on personal communication with the authors.

We first prove that 1) implies 2). As mentioned in the introduction, an interactive communication $\mathbf{F}$ without private randomization suffices to achieve positive key rates, i.e., without loss of generality we may assume $H(F_1|X^n)=H(F_2|F_1Y^n)=H(F_3|F_{1:2}X^n)=...=0$. Orlitsky and Wigderson observe that for any $\mathbf{F}$ without private randomization, the conditional pmf of $(X^n, Y^n, Z^n)$ given $\mathbf{F}=\mathbf{f}$ has the form $p_r(x^n, y^n, z^n)$ for some sets $\mathcal{A}$ and $\mathcal{B}$ that depend on $\mathbf{f}$. This follows from the rectangle property of communication complexity, e.g., see \cite[p. 10]{kushilevitz_nisan_1996}. Given $S(X;Y\|Z)>0$, from 1) there is an $(n,\delta)$ code with a positive rate $R>0$ and a sufficiently small $\delta$. Observing that $\frac 1n I(K_A;Y^n|\mathbf{F})\geq \frac 1n I(K_A;K_B|\mathbf{F})\approx R$  and $\frac 1n I(K_A;Z^n|\mathbf{F})\approx 0$, we have
\begin{align}I(K_A;Y^n|\mathbf{F})-I(K_A;Z^n|\mathbf{F})>0\end{align}
so there exists an $\mathbf{F}=\mathbf{f}$ such that 
\begin{align}I(K_A;Y^n|\mathbf{f})-I(K_A;Z^n|\mathbf{f})>0.\end{align}
Note that given a $\mathbf{F}=\mathbf{f}$, $K_A$ is a function of $X^n$. 

We next prove that 2) implies 3). Suppose that Alice and Bob observe $N$ independent blocks, each of which consists of $n$ i.i.d.\ realizations of $(X,Y)$. Consider one of the blocks. Suppose that Alice observes $X^n$ and Bob observes $Y^n$ in that block. Alice declares on the public channel whether or not $X^n\in\mathcal{A}$, and Bob declares whether or not $Y^n\in\mathcal{B}$.  If $X^n\notin\mathcal{A}$ or $Y^n\notin\mathcal{B}$, they discard the block. Otherwise, they use the block for key agreement. The fraction of used blocks is asymptotically equal to $\mathbb{P}[X^n\in \mathcal{A}, Y^n\in \mathcal{B}]$. In the used blocks, the conditional pmf of the source is $p_r(x^n, y^n, z^n)$, and Alice has a simple entropic advantage over Eve, so she can apply a one-way SK generation scheme to produce a shared key with Bob. Alice's and Bob's declarations count as two rounds of communications. However, if Alice declares that $X^n\in\mathcal{A}$ and if $Y^n\in\mathcal{B}$, then Bob hears Alice's response and attaches the necessary public information for SK generation to his declaration so that no more than two rounds of communication are required. 

Since 1) implies 2) and 2) implies 3), we have that 1) also implies 3). Finally, going from 3) to 1) is immediate, and going from 3) to 2) is possible by going from 3) to 1) and then 1) to 2), as shown above.
\end{proof}

Finally, a sufficient condition for $S(X;Y\|Z)=0$ for the special class of erasure sources is given in \cite{MaurerWolf99}.
\begin{definition}\cite[Definition 4]{MaurerWolf99}\label{def:Fpxy}
Given a joint pmf $p_{XY}$ on discrete sets $\mathcal{X}$ and $\mathcal{Y}$, let 
\begin{align}
&F(p_{XY})=\min_{q_{XY}}\left(\max_{x,y}\left(\frac{p_{XY}(x,y)}{q_{XY}(x,y)}\right)\cdot\max_{x,y}\left(\frac{q_{XY}(x,y)}{p_{XY}(x,y)}\right)\right)
\end{align}
where the minimum is taken over all product pmfs $q_{XY}=q_X q_Y$ and where we set $0/0:=1$, $c/0:=\infty$ for $c>0$. 
Further, the \emph{``deviation from independence of $p_{XY}$"} is defined as
\begin{align}
   d_{\emph{ind}}(p_{XY})=1-\frac{1}{F(p_{XY})}.
\end{align}
\label{def-Maurer}
\end{definition}
For example, for the DSBE we have \cite[p. 509]{MaurerWolf99}
\begin{align}
	d_{\text{ind}}(p_{XY})= 1- \frac{p}{1-p}. 
\end{align}

\begin{theorem} \cite[Theorems 14 and 15]{MaurerWolf99}\label{theo:MaurerWolftheo1415}
For the erasure source $p_{XY}(x,y)p_{Z|XY}(z|x,y)$ with erasure probability $\epsilon$, we have $S(X;Y\|Z)=B_0(X;Y\|Z)=0$ if 
$\epsilon\leq 1-d_{\emph{ind}}(p_{XY})$. Furthermore, for the special case of the DSBE source the converse is also true, namely $S(X;Y\|Z)>0$ if $\epsilon>1-d_{\emph{ind}}(p_{XY})$. \label{theoremWolf} 
\end{theorem}
In this paper, we generalize Theorem~\ref{theo:MaurerWolftheo1415} with Theorem \ref{theoremeq:eps12} below.

\subsection{Classification of Sources with Zero Secret Key Capacity}
We introduce a new quantity $\Delta(X;Y\|Z)$ that provides insight into sources with zero key capacity. In particular, it allows one to compare sources with zero key capacity with each other. When $S(X;Y\|Z)=0$, it is not possible for Alice and Bob to agree on a shared key. However, it may be still possible for Alice and Bob to agree on a key that is \emph{approximately} secret. The quantity $\Delta(X;Y\|Z)$ quantifies the goodness of the approximate secret key.

Suppose Alice and Bob wish to agree on a \emph{single secret bit}. That is, instead of achieving a key \emph{rate}, they produce \emph{bits} $K_A\in\{1,2\}$ and $K_B\in\{1,2\}$ respectively. Alice and Bob wish to maximize $\mathbb{P}[K_A=K_B]$ while minimizing leakage to Eve who has $Z^n$ and the public discussion transcript $\mathbf{F}$. We can measure the quality of the keys $K_A$ and $K_B$ via the total variation distance
\begin{align}
\|p_{K_A K_B Z^n \mathbf{F}} - q_{K_AK_B}\cdot p_{Z^n \mathbf{F}}  \|_{TV} \label{quality}
\end{align}
where \begin{align}q_{K_AK_B}(k_A, k_B)=\frac{1}{2}\mathds{1}[k_A=k_B]\label{quality-q}\end{align} is the ideal distribution on $\mathcal K_A\times \mathcal K_B=\{1,2\}^2$. If the total variation distance given in \eqref{quality} vanishes, then Alice and Bob have perfect secret bits. The same total variation distance as in \eqref{quality-q} has been previously utilized in \cite[Eq. 3]{tyagi2014bound}.

We are interested in the  minimum of \eqref{quality} over all public discussion protocols as we let the number of source observations $n$ tend to infinity.

\begin{definition} Given a source $p_{XYZ}$, let $\Delta(X;Y\|Z)$ be the infimum of (\ref{quality}) over all public discussion schemes $\mathbf{F}$ (of arbitrary length) that produce single bits $K_A$ and $K_B$ by Alice and Bob, respectively, and where the number $n$ of source observations tends to infinity. In other words, we let
\begin{align}
   &\Delta(X;Y\|Z)=\inf \|p_{K_A K_B Z^n \mathbf{F}} - q_{K_AK_B}\cdot p_{Z^n \mathbf{F}}  \|_{TV}
\end{align}
where the infimum is over all $n\in\mathbb{N}$, arbitrary finite sets $\mathcal{F}_1$, $\mathcal{F}_2, \cdots, \mathcal{F}_k$ for some $k\in\mathbb{N}$, arbitrary conditional pmfs $p_{F_1|X^n}, 
p_{F_2|F_1Y^n}, p_{F_3|F_1F_2X^n},\ldots,$ and arbitrary conditional pmfs $p_{K_A|X^n\mathbf{F}}$ and $p_{K_B|Y^n\mathbf{F}}$, where $K_A$ and $K_B$ are binary and $q_{K_AK_B}$ is as given in \eqref{quality-q}. 
\end{definition}

\section{Main Results}\label{main-result}

We give generic results about the positivity of the SK capacity in Section \ref{sec:general}. We then restrict attention to erasure sources in Section \ref{sec:erasure}. 

\subsection{General Sources}\label{sec:general}
\begin{theorem}\label{genertlsd} Let $p_{XY}$ and $q_{XY}$ be two pmfs satisfying $q_{XY}\preceq p_{XY}$ (as defined in Definition \ref{defpreceq}). Consider a channel $p_{Z|XY}$ and let 
\begin{align*}
	&p_{XYZ}=p_{XY}\, p_{Z|XY} \\ &q_{XYZ} = q_{XY}\, p_{Z|XY}.
\end{align*}	
If the SK capacity $S(X;Y\|Z)$ under $q_{XYZ}$ is positive, then the SK capacity $S(X;Y\|Z)$ under $p_{XYZ}$ is also positive.
\end{theorem}
Theorem~\ref{genertlsd}  is proved in Section~\ref{eq:genertlsd}. Intuitively, the condition $q_{XY}\preceq p_{XY}$ allows simulating the source $q_{XY}$ from $p_{XY}$ via rejection sampling, i.e., each observation $(X,Y)$ from $p_{XY}$ is either accepted or rejected by Alice and Bob via discussion on the public channel. 
If it is possible to generate a secret key under $q_{XYZ}$, then Alice and Bob simulate $q_{XYZ}$ from $p_{XYZ}$ and utilize the protocol for $q_{XYZ}$ to generate a SK with positive rate.

The following theorem gives a condition to achieve a positive SK rate in terms of the R{\'e}nyi divergence of order $1/2$.

\begin{theorem}\label{gentheoremeq:eps12} Consider the source $p_{XY}\,p_{Z|XY}$. Then the following conditions are equivalent.
\begin{enumerate}[label=(\roman*)]
\item  The key capacity $S(X;Y\|Z)$ is positive.
\item  There is an integer $n$, disjoint non-empty sets $\mathcal{A}_1, \mathcal{A}_2\subset \mathcal{X}^n$, and disjoint non-empty sets $\mathcal{B}_1, \mathcal{B}_2\subset \mathcal{Y}^n$ such that (see Definition~\ref{def:renyidivv}) for $(X^n, Y^n, Z^n)$ that are i.i.d. with pmf $p_{X YZ}$ we have
\end{enumerate}
\begin{align}
	&D_{\frac 12}\Big(p_{Z^n}(\cdot | X^n\!\in\!\mathcal{A}_1,Y^n\!\in\!\mathcal{B}_1)\Big\| p_{Z^n}(\cdot  |X^n\!\in\!\mathcal{A}_2,Y^n\!\in\!\mathcal{B}_2)\Big)\nonumber\\
	&<
	\log\Bigg(\frac{
	\mathbb{P}[X^n\!\in\!\mathcal{A}_1,Y^n\!\in\!\mathcal{B}_1]
}{
\mathbb{P}[X^n\!\in\!\mathcal{A}_1,Y^n\!\in\!\mathcal{B}_2]
}\frac{
	\mathbb{P}[X^n\!\in\!\mathcal{A}_2,Y^n\!\in\!\mathcal{B}_2]
}{
\mathbb{P}[X^n\!\in\!\mathcal{A}_2,Y^n\!\in\!\mathcal{B}_1]
}\Bigg)\label{eq:nproductnonzeroSKrate}
\end{align}
where $p_{Z^n}(\cdot | \mathcal{E})$ is the pmf of $Z^n$ conditioned on the event $\mathcal{E}$.	
\begin{enumerate}[label=(\roman*)]
\setcounter{enumi}{2}
\item $\Delta(X;Y\|Z)=0$.
\item $\Delta(X;Y\|Z)<\frac{3-\sqrt{5}}{8}\approx 0.095$.
\end{enumerate}
\end{theorem}
\begin{remark} Both the characterization given in item 2 of Theorem \ref{OrlitskyWigdersonThm} and  the characterization in item (ii) of Theorem \ref{gentheoremeq:eps12} are $n$-letter characterizations. Neither  makes the problem of identifying sources with positive SK rate decidable for general sources, \emph{i.e.,} the characterizations are not computable. However, for 
the special case of erasure sources when  either $X$ or $Y$ is binary, the characterization in Theorem \ref{gentheoremeq:eps12} turns out to be helpful, while it is not clear how to make use of the characterization in Theorem \ref{OrlitskyWigdersonThm}. Moreover,   in contrast to the characterization given in item 2 of Theorem \ref{OrlitskyWigdersonThm}, the characterization in Theorem \ref{gentheoremeq:eps12} considers all probabilities with respect to the unconditional product pmf $\prod_{i=1}^n p_{XYZ}(x,y,z)$. Furthermore, \eqref{eq:nproductnonzeroSKrate} can be equivalently expressed as
\begin{align}\nonumber
	&\sum_{z^n}\!\Big(\mathbb{P}[Z^n=z^n, X^n\!\in\!\mathcal{A}_1,Y^n\!\in\!\mathcal{B}_1]^{\frac{1}{2}}\times \mathbb{P}[Z^n=z^n, X^n\!\in\!\mathcal{A}_2,Y^n\!\in\!\mathcal{B}_2]^{\frac 12}\Big)\nonumber
	\\[0.5ex]&\;>
\mathbb{P}[X^n\!\in\!\mathcal{A}_1,Y^n\!\in\!\mathcal{B}_2]^{\frac 12}\;
\mathbb{P}[X^n\!\in\!\mathcal{A}_2,Y^n\!\in\!\mathcal{B}_1]^{\frac 12}
\label{eq:nproductnonzeroSKrate23}.
\end{align}
Equation \eqref{eq:nproductnonzeroSKrate23} involves only a product of unconditional probability terms, while  the characterization given in item 2 of Theorem \ref{OrlitskyWigdersonThm} is based on mutual information for conditional expressions. 
\end{remark}
\begin{remark} 
The constant $\frac{3-\sqrt{5}}{8}$ in Theorem \ref{gentheoremeq:eps12} is not necessarily optimal. 
It is an interesting question to find the best possible constant, i.e., the
minimum possible value of $\Delta(X;Y\|Z)$ over all sources that satisfy $S(X;Y\|Z)=0$.
\end{remark}
\begin{corollary}\label{corollary-new-1} Considering item (ii) of Theorem~\ref{gentheoremeq:eps12} for $n=1$, the SK capacity $S(X;Y\|Z)$ is positive  if one can find distinct symbols $x_1, x_2 \in\mathcal{X}$ and distinct symbols $y_1, y_2 \in\mathcal{Y}$ such that 
\begin{align}
	&D_{\frac 12}(p_{Z|XY}(\cdot |x_1,y_1)\| p_{Z|XY}(\cdot |x_{2},y_{2}))<
	\log\left(\frac{
		p_{XY}(x_1, y_{1})p_{XY}(x_2, y_{2})}{p_{XY}(x_1, y_{2})p_{XY}(x_2, y_{1})}\right).
\label{eqn:gentheoremeqA2}
\end{align}	

\begin{corollary} 
Fix some pmfs $q_{X,1}$, $q_{X,2}$, $q_{Y,1}$, $q_{Y,2}$, and $q_Z$. 
For $i,j=1,2$, we define
\begin{align}
  \kappa_{i,j}&=\min_{r_{XY}:\: r_X=q_{X,i}, r_Y=q_{Y,j}}D(r_{XY}\|p_{XY})
\end{align}
and
\begin{align}
\theta_{i,j}&=\min_{r_{XYZ}} D(r_{XYZ}\|p_{XYZ}) \label{eqnDR}
\end{align}
where the minimum in \eqref{eqnDR} is over $r_{XYZ}$ satisfying $r_X=q_{X,i}, r_Y=q_{Y,j}, r_Z=q_Z$. 
Then, $S(X;Y\|Z)>0$  if  
\begin{align}&\theta_{1,1}+\theta_{2,2}<
\kappa_{1,2}+\kappa_{2,1}.\label{cond1}
\end{align}
To see this, assume that $q_Z$, $q_{X,1}$, $q_{X,2}$, $q_{Y,1}$, and $q_{Y,2}$ are types, and $\mathcal{A}_i$ and $\mathcal{B}_j$ are the set of typical sequences with types $q_{X,i}$ and $q_{Y,j}$, respectively, \emph{i.e.,} $\mathcal{A}_i=\mathcal{T}^{(n)}_{q_{X,i}}$ and $\mathcal{B}_j=\mathcal{T}^{(n)}_{q_{Y,j}}$. 

Positivity of the SK capacity follows if \eqref{eq:nproductnonzeroSKrate23} holds. Using Lemma \ref{LemmaMT} and the simple inequality
\begin{align}\nonumber
	&\sum_{z^n}\!\Big(\mathbb{P}[Z^n=z^n, X^n\!\in\!\mathcal{A}_1,Y^n\!\in\!\mathcal{B}_1]^{\frac{1}{2}}\times \mathbb{P}[Z^n=z^n, X^n\!\in\!\mathcal{A}_2,Y^n\!\in\!\mathcal{B}_2]^{\frac 12}\Big)\nonumber
	\\[0.5ex]&\;\geq 
\sum_{z^n\in \mathcal{T}^{(n)}_{q_Z}}\!\Big(\mathbb{P}[Z^n=z^n, X^n\!\in\!\mathcal{A}_1,Y^n\!\in\!\mathcal{B}_1]^{\frac{1}{2}}\times \mathbb{P}[Z^n=z^n, X^n\!\in\!\mathcal{A}_2,Y^n\!\in\!\mathcal{B}_2]^{\frac 12}\Big)
\end{align}
we observe that \eqref{eq:nproductnonzeroSKrate23} holds as $n\rightarrow \infty$ if \eqref{cond1} holds. 

Consider an erasure source, and the special case of $q_Z(z)=\mathbf{1}[z=e]$. Then we have
\begin{align}
\theta_{i,j}=\kappa_{i,j}+\log\frac{1}{\epsilon}.
\end{align}
In other words, the SK rate is positive if 
\begin{align}&\log(\epsilon)>
\frac{1}{2}\left(\kappa_{1,1}+\kappa_{2,2}-\kappa_{1,2}-\kappa_{2,1}\right).
\end{align}
In particular if $q_{X,1}(x)=\mathbf{1}[x=x_1]$, $q_{X,2}(x)=\mathbf{1}[x=x_2]$, $q_{Y,1}(y)=\mathbf{1}[y=y_1]$, $q_{Y,2}(y)=\mathbf{1}[y=y_2]$, we obtain
\begin{align}
  \kappa_{i,j}=\log\frac{1}{p_{XY}(x_i, y_j)}
\end{align}
and hence the SK key is positive if
\begin{align}
  \epsilon>\left(\frac{p_{XY}(x_1, y_2) p_{XY}(x_2, y_1)}{p_{XY}(x_1,y_1)p_{XY}(x_2,y_2)}\right)^{\frac 12}.\label{eq:conditionforerasuresourcenonzerorate}
\end{align}
Alternatively, the condition (\ref{eq:conditionforerasuresourcenonzerorate}) can be deduced also from \eqref{eqn:gentheoremeqA2} (see \eqref{eqnN149} for the calculation). 
\end{corollary}
\end{corollary}

Theorem~\ref{gentheoremeq:eps12} is proved in Section~\ref{proofgentheoremeq:eps12} by using a hypothesis testing approach. The left hand side of \eqref{eq:nproductnonzeroSKrate} is the error exponent of the adversary in a hypothesis testing problem while the right hand side is the error exponent of the legitimate parties. The theorem shows that key agreement is feasible when the legitimate parties have a better exponent than the adversary.

\subsection{Erasure Sources}\label{sec:erasure}
We illustrate the condition (\ref{eq:nproductnonzeroSKrate}) for erasure sources and relate it to previously known bounds. Suppose we are given a joint pmf $p_{XY}$. Without loss of generality, we assume that $p_X(x)>0$ and $p_Y(y)>0$ for all $(x,y)\in \mathcal{X}\times\mathcal{Y}$ throughout this section. 
We define a path, which is used in the proofs of the theorems given below.
\begin{definition}\label{def:path}
A sequence $(x_1, y_1, x_2, y_2, \cdots, x_k, y_k)$ forms a path if all $x_i$'s with $x_i\in\mathcal{X}$ are distinct and also all $y_i$'s with $y_i\in\mathcal{Y}$ are distinct. We say the length of the path is $2k$ and we assign the following value to the path 
\begin{align}
	\left(\frac{\prod_{i=1}^kp_{XY}(x_i,y_i)}{p_{XY}(x_1, y_k) \prod_{i=2}^kp_{XY}(x_i,y_{i-1})}\right)^{1/k}.\label{eq:assignedvalue}
\end{align}
Let $\epsilon_{1}$ be the minimum  value of all possible paths and $\epsilon_{2}$ be the minimum  value of all possible paths of length at most four. In particular, we have
\begin{align}
	\epsilon_2=\min_{x_1\neq x_2, y_1\neq y_2}\left(\frac{p_{XY}(x_1,y_1)p_{XY}(x_2,y_2)}{p_{XY}(x_1, y_2) p_{XY}(x_2, y_1)}\right)^{\frac 12}.
\end{align}
\end{definition}

\begin{example} \label{example2n}
\normalfont	
Suppose that $X$ and $Y$ are binary with a joint pmf $p_{XY}$. Then paths of length two are of the form $(x_1, y_1)$ for some $x_1,y_1\in\{0,1\}$, and, by definition, are assigned the value $1$. Because the $x_i$'s and $y_i$'s are distinct in a path, paths of length more than four do not exist. There are multiple paths of length four. For instance, $(x_1=0, y_1=1, x_2=1, y_2=0)$ is assigned the value
$$\frac{\sqrt{p_{XY}(0,1)p_{XY}(1,0)}}{\sqrt{p_{XY}(0,0)p_{XY}(1,1)}}$$
and $(x_1=0, y_1=0, x_2=1, y_2=1)$ is assigned the value
$$\frac{\sqrt{p_{XY}(0,0)p_{XY}(1,1)}}{\sqrt{p_{XY}(0,1)p_{XY}(1,0)}}.$$ 
All other paths have one of the above two values and therefore
\begin{align}
&\epsilon_1=\epsilon_2=\min\Bigg\{\frac{\sqrt{p_{XY}(0,1)p_{XY}(1,0)}}{\sqrt{p_{XY}(0,0)p_{XY}(1,1)}},\frac{\sqrt{p_{XY}(0,0)p_{XY}(1,1)}}{\sqrt{p_{XY}(0,1)p_{XY}(1,0)}}\Bigg\}.\label{eqn:binbi}
\end{align}
Note that one of the terms inside the minimum is less than or equal to one. We therefore do not need to consider paths of length two whose values are one.      
\end{example}
\begin{example} \label{example3n}
As mentioned at the beginning of this section, we assume positive marginal distributions $p_X(x)>0$ and $p_Y(y)>0$ for all $(x,y)\in \mathcal{X}\times\mathcal{Y}$.
Assume that $p_{XY}(x^*,y^*)=0$ for some $(x^*,y^*)\in\mathcal{X}\times\mathcal{Y}$.   Then $\epsilon_1=\epsilon_2=0$. This can be seen by starting the path with $x_1=x^*, y_1=y^*$.
\end{example}

We now give lower and upper bounds on the maximum erasure probability for which the SK capacity is zero for an erasure source.

\begin{theorem}\label{theoremeq:eps12} For the erasure source $p_{XY}\, p_{Z|XY}$ with erasure probability $\epsilon$, we have $S(X;Y\|Z)=0$ if $\epsilon\leq \epsilon_1$, and $S(X;Y\|Z)>0$ if $\epsilon> \epsilon_2$, where $\epsilon_1$ and $\epsilon_2$ are as in Definition~\ref{def:path}. Moreover, $\epsilon_1=\epsilon_2$ if $X$ or $Y$ is binary. We also have $\epsilon_1=\epsilon_2=0$ when $p_{XY}(x^*,y^*)=0$ for some $(x^*,y^*)\in\mathcal{X}\times\mathcal{Y}$. For these special cases, $S(X;Y\|Z)>0$ if and only if $\epsilon>\epsilon_1=\epsilon_2$.
\label{mmthm1s}
\end{theorem}

Theorem~\ref{theoremeq:eps12} is proved in Section~\ref{eq:maintheoremsproof}. Positivity of $S(X;Y\|Z)$ for $\epsilon> \epsilon_2$ is derived by using the result of Theorem \ref{gentheoremeq:eps12}. 

\begin{remark} Parameter $1-\epsilon$ quantifies the information leakage to Eve. Intuitively speaking, the condition $\epsilon>\epsilon_2$ or $1-\epsilon<1-\epsilon_2$ states that key agreement is possible if the correlation between $X$ and $Y$ (as measured by $1-\epsilon_2$) is larger than the leakage to Eve. In fact, the quantities $\epsilon_1$ and $\epsilon_2$ can be used to define measures of correlation. For instance, Maurer and Wolf take 
\begin{align}
	d_{\emph{ind}}(p_{XY})=1-\frac{1}{F(p_{XY})}=1-\epsilon_1
\end{align}	
as a measure of correlation. We propose 
\begin{align}
	\log F(p_{XY})=\log\frac{1}{\epsilon_1}
\end{align}
as yet another measure. Observe that $\log F(p_{XY})$ can be expressed in terms of Renyi-divergence of order infinity:
\begin{align}
	\log\frac{1}{\epsilon_1}&=\log F(p_{XY})=\min_{q_X,q_Y}\!\bigg(D_\infty(p_{XY} \| q_X\, q_Y)+D_\infty(q_X\,q_Y \| p_{XY})\bigg).
\end{align}
We define the Renyi-Jeffrey's divergence (RJ divergence) between two distributions $p$ and $q$ as
\begin{align}
	D^{RJ}_{\alpha}(p\|q)\triangleq \frac{1}{2}\left(D_\alpha(p\|q)+D_\alpha(q\|p)\right)
\end{align}
which is just Jeffrey's divergence (symmetrized KL divergence) in its R{\'e}nyi form. We next define the RJ information of order $\alpha$ as the minimum RJ divergence between a given joint distribution and all product distributions:
\begin{align}
	RJ_\alpha(X;Y)\triangleq \min_{q_Xq_Y}D^{RJ}_{\alpha}(p_{XY} \| q_X\, q_Y).\label{eq:RJalpha}
\end{align}
Observe that $RJ_\infty(X;Y)=\frac12\log F(p_{XY})$. The way RJ information is defined in (\ref{eq:RJalpha}) parallels the way $\alpha$-R{\'e}nyi mutual information is defined in \cite{lapidoth2019two} and \cite{tomamichel2017operational}[Equation (58)].

Similarly for $\epsilon_2$, we propose 
\begin{align}
	\log\frac{1}{\epsilon_2}
\end{align}
as a new measure of correlation (see (\ref{eqn:J134}) below) and study its properties in Appendix \ref{appendixB}. 
\end{remark}
\begin{remark}
Theorem~\ref{theoremeq:eps12} generalizes \cite[Theorems 14 and 15]{MaurerWolf99} which was summarized in Theorem \ref{theoremWolf}. In the proof, we show (via the duality theorem for linear programs) that $\epsilon_1$ (of Definition \ref{def:path}) is the same quantity as $1-d_{\emph{ind}}(p_{XY})$ (of Definition \ref{def-Maurer}). Observe that Theorem~\ref{theoremeq:eps12} claims $S(X;Y\|Z)>0$ if $\epsilon> \epsilon_2$ for any pmf $p_{XYZ}$, while \cite[Theorem 14]{MaurerWolf99} considers only the DSBE source. We remark that (i) the code we use to prove $S(X;Y\|Z)>0$ for general distributions differs from the one used by \cite{MaurerWolf99} for the DSBE source. Our code applies the swapping concept  and works for general sources, (ii) for the special case of the DSBE source, our code and the one used in \cite{MaurerWolf99} give the same bound on $\epsilon$ for the positivity of the SK capacity, 
(iii) for the DSBE source, the code used in \cite{MaurerWolf99} achieves higher secret key rates for $\epsilon>\epsilon_2$. However, the code of \cite{MaurerWolf99} exchanges more information on the public channel.
\end{remark}

We now study the one-way SK rate and the lower bound $\bar{L}(X;Y\|Z)$ obtained from (\ref{eq:lowerbound}) for an erasure source.

\begin{theorem}\label{theorem-lowerbounds} For an erasure source $p_{XY}\,p_{Z|XY}$ with erasure probability $\epsilon$ such that $p_{X}(x)>0,\; p_Y(y)>0$ for all $(x,y)\in\mathcal{X}\times\mathcal{Y}$, the following statements hold.
\begin{enumerate}
\item The  one-way SK rate from Alice to Bob vanishes if and only if
\begin{align}
	\epsilon\leq 1-\eta(p_{Y|X})\label{eq:Theorem6part1}
\end{align}	
where $\eta(\cdot)$ is defined in (\ref{eqnAG687rr657}). A similar statement holds for the one-way SK rate from Bob to Alice. 
\item We have $\bar{L}(X;Y\|Z)=0$ if and only if 
\begin{align}
	\epsilon\leq 1-\max_{\substack{q_{XY}:\: q_{XY}\preceq p_{XY}}}\rho^2_m(q_{XY})\label{eq:theorem4lowerbound}
\end{align}
where $\rho_m(\cdot)$ is defined in (\ref{eq:max-correlatoin-31}). 
\item The upper bound $ B_0(X;Y\|Z)$ in (\ref{eq:B0upperbound}) is zero  if and only if $\epsilon\leq \epsilon_3$ where
\begin{align}
	\epsilon_3=\max \sum_{t=1}^{|\mathcal{Z}|}\min_{x,y:\: p_{X,Y}(x,y)>0}\delta_{x,y,t} \label{theoremeq:maxmin2}
\end{align} 
and the maximum is over all $\delta_{x,y,t}$ such that $\delta_{x,y,t}\geq 0$, the matrix $[p_{XY}(x,y)\delta_{x,y,t}]$  has rank 1 for all $t$, and $\sum_{t=1}^{|\mathcal{Z}|}\delta_{x,y,t}=1$ for all $x,y$. Here, for every $t$, $[p_{XY}(x,y)\delta_{x,y,t}]$  is
a matrix with dimensions ${|\mathcal{X}|\times|\mathcal{Y}|}$ whose rows  and columns are indexed by the realizations of $X$ and $Y$, respectively, and whose $(x,y)$ entry is $p_{XY}(x,y)\delta_{x,y,t}$ for all $(x,y)\in\mathcal{X}\times\mathcal{Y}$.
\item Assume that $p(x,y)>0$ for all $x,y$ (the case where $p(x,y)=0$ for some $x,y$ was discussed in Theorem \ref{theoremeq:eps12}). The upper bound $ B_1(X;Y\|Z)$ in (\ref{eq:B1upperbound}) is zero  if and only if 
$\epsilon\leq \epsilon_4$ where
\begin{align}
	\epsilon_4 =1- \inf\eta(p_{J|XY}).
\end{align}
Here the infimum is taken over channels $p_{J|XY}$ 
for which $I(X;Y|J)=0$ for $p_{XYJ}=p_{XY}p_{J|XY}$. 
\end{enumerate}

\end{theorem}

The proof of Theorem~\ref{theorem-lowerbounds} is given in Section~\ref{subsec:proofdsbe22}.
\begin{remark}
Since $B_1(X; Y ||Z) \!\leq\! B_0(X; Y ||Z)$, we have $\epsilon_3 \leq\epsilon_4$.
We show below that $\epsilon_1 \leq \epsilon_3$, where $\epsilon_1$ is as given in Definition \ref{def:path}.
This implies that $\epsilon_1 \leq \epsilon_3 \leq \epsilon_4$.
Therefore, the bound in terms of $\epsilon_4$ given in the last part of Theorem \ref{theorem-lowerbounds} is tighter than the bound in terms of $\epsilon_1$ given in Theorem \ref{theoremeq:eps12}.
The reason for stating Theorem \ref{theoremeq:eps12} with $\epsilon_1$ instead of $\epsilon_4$ is that the definition of $\epsilon_1$ is explicit, computable and can be readily related to $\epsilon_2$.
On the other hand, $\epsilon_4$ is not computable in general since we do not have a cardinality bound on $J$.
However, for any particular choice of $p_{J|XY}$ we deduce that $S(X;Y\|Z)=0$ if  $I(X;Y|J)=0$ and $\epsilon\leq 1- \eta(p_{J|XY})$, but we do not know the best value for $1- \eta(p_{J|XY})$ as we vary over all $p_{J|XY}$ satisfying $I(X;Y|J)=0$.

The value of $\epsilon_3$ given in part 3 of Theorem \ref{theorem-lowerbounds} is greater than or equal to $\epsilon_1$ given in Definition \ref{def:path}. To see this,
observe from Theorem \ref{lemma-generaldis1} that one can find some $\tilde{ \delta}_{x,y}\in [0,1]$ such that 
\begin{align}
	\epsilon_1= \min_{x,y}\tilde{\delta}_{x,y}
\end{align}
and the matrix $[p_{XY}(x,y)\tilde{\delta}_{x,y}]$ has rank one. To prove that $\epsilon_1\leq \epsilon_3$,  we need to find appropriate $\delta_{x,y,t}$ for $t=1,2,\cdots, |\mathcal{Z}|$ such that 
\begin{align}
	\min_{x,y}\tilde{\delta}_{x,y}\leq \sum_{t=1}^{|\mathcal{Z}|}\min_{x,y:\: p_{XY}(x,y)>0}\delta_{x,y,t}.
\end{align}

As shown in the proof of part 3 of Theorem \ref{theorem-lowerbounds}, the quantity $\epsilon_3$ remains the same if we allow  $t$ to take values in a larger set $\{1,2,\cdots, M\}$ for some $M>|\mathcal{Z}|$. We define $\delta_{x,y,t}$ for 
$t\in\{0\}\cup \mathcal{X}\times\mathcal Y$ as follows: $\delta_{x,y,0}=\tilde{\delta}_{x,y}$ and for any $(x',y')\in \mathcal{X}\times\mathcal Y$, we have
\begin{align}
	\delta_{x,y, t=(x',y')}=\mathds{1}[x'=x, y'=y](1-\tilde{\delta}_{x,y}).
\end{align}
We have $\sum_t \delta_{x,y,t}=1$ and $[p_{XY}(x,y)\delta_{x,y,t}]$  has rank 1 for all $t\in\{0\}\cup \mathcal{X}\times\mathcal Y$. Furthermore, we have
\begin{align}
	\min_{x,y: p_{XY}(x,y)>0}\delta_{x,y,t=(x',y')}=0 \qquad \forall x',y'.
\end{align}
Therefore, we compute
\begin{align}
	&\sum_{t}\min_{x,y:\: p_{X,Y}(x,y)>0}\delta_{x,y,t}=\min_{x,y:\: p_{X,Y}(x,y)>0}\delta_{x,y,0}=\min_{x,y:\: p_{X,Y}(x,y)>0}\tilde{\delta}_{x,y}\geq \min_{x,y}\tilde{\delta}_{x,y}.
\end{align}
\end{remark}
\vspace{0.4cm}

Computing the bounds on $\epsilon$ given in Theorem \ref{theorem-lowerbounds}  is cumbersome for general distributions. Thus, we next focus on the DSBE source and illustrate the suboptimality of the lower bound $\bar{L}(X;Y\|Z)$ obtained from (\ref{eq:lowerbound}) with a DSBE source example.

\subsection{DSBE Source Example}
Using (\ref{eqn:binbi}) and Theorem~\ref{theoremeq:eps12}, the SK capacity $S(X;Y\|Z)$ is zero if and only if 
\begin{align}
\epsilon\leq\frac{\min\{p,1-p\}}{\max\{p,1-p\}}.\label{eqfger2}
\end{align}
We now study the  lower bound $\bar{L}(X;Y\|Z)$ in (\ref{eq:lowerbound}). The main result of this subsection is to show that $S(X;Y\|Z)\neq \bar{L}(X;Y\|Z)$ for a DSBE$(p,\epsilon)$ source if
\begin{align}\frac{\min\{p,1-p\}}{\max\{p,1-p\}}<\epsilon\leq 4p(1-p).\end{align}
In fact, we show that $\bar{L}(X;Y\|Z)=0$ for erasure probabilities $\epsilon$ in the above interval since we know from \eqref{eqfger2} that $S(X;Y\|Z)>0$ in this interval. This result illustrates that the {lower bound $\bar{L}(X;Y\|Z)$} is loose.

We remark that the lower bound is tight, i.e., $S(X;Y\|Z)=\bar{L}(X;Y\|Z)$, for all previously considered joint pmfs $p_{XYZ}$ for which the SK capacity $S(X;Y\|Z)$ is known. For instance, if $X\rightarrow Y\rightarrow Z$ forms a Markov chain, then assigning $U_1=X$ and $k=1$ in (\ref{eq:lowerbound}) recovers the SK capacity $S(X;Y\|Z)=I(X;Y|Z)$ achieved by one-way communication from $X$ to $Y$. Similarly, consider the reversely degraded example from \cite{AhlswedeCsiszar}. Let $X = (X_1, X_2)$, $Y = (Y_1, Y_2)$, and $Z = (Z_1, Z_2)$, where all $(X_i,Y_i,Z_i)$ tuples for $i=1,2$ are mutually independent. If $X_{1}\rightarrow Y_{1}\rightarrow Z_{1}$ and $Y_{2}\rightarrow X_{2}\rightarrow Z_{2}$ form Markov chains, assigning $U_1=X_1$ and $U_{2}=Y_{2}$ in (\ref{eq:lowerbound}) recovers the SK capacity $S(X;Y\|Z)=I(X;Y|Z)$.

We give the condition for $\bar{L}(X;Y\|Z)=0$ for a DSBE source in the following theorem and prove it in Section~\ref{subsec:proofdsbe}.

\begin{theorem}\label{thm:A5}
 Let $(X,Y,Z)$ be a DSBE  source with parameters $(p,\epsilon)$. Then $\bar{L}(X;Y\|Z)=0$ if and only if the one-way SK rate from Alice to Bob (or Bob to Alice) vanishes, i.e., if and only if
\begin{align}\epsilon\leq 4p(1-p).\end{align}
\end{theorem}

\begin{figure*}[t!]
	\centering
	\includegraphics[width=0.85\textwidth, height=1\textheight, keepaspectratio]{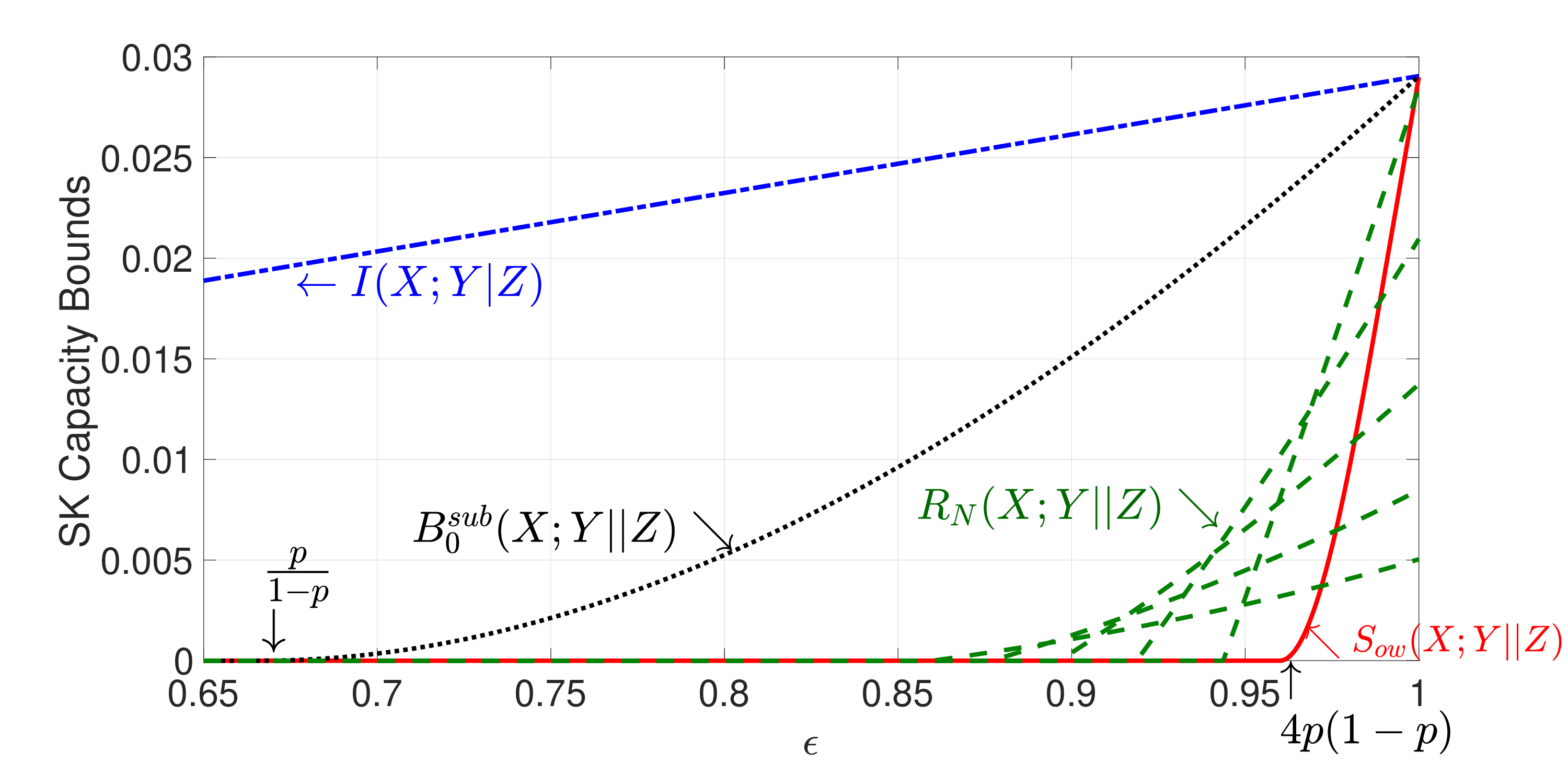}
	\caption{SK capacity bounds for a DSBE$(0.4,\epsilon)$ source. The curve labeled $S_{\text{ow}}(X;Y\|Z)$ is zero if and only if $\epsilon\leq 4p(1-p)$.}\label{fig:SK06}
\end{figure*} 

In Fig.~\ref{fig:SK06}, we plot the known lower and upper bounds on $S(X;Y\|Z)$ to illustrate the gaps between them. Consider a DSBE source $(X,Y,Z)$ with parameters $(p=0.4,\epsilon)$. This source has
\begin{align}
  I(X;Y)=\frac{I(X;Y|Z)}{\epsilon}.
\end{align}
Therefore, we plot only $I(X;Y|Z)$ and do not consider the upper bound $I(X;Y)$. We also plot the improved upper bound (see (\ref{eq:B0upperbound}))
\begin{align}B_0^{\text{sub}}(X;Y\|Z)=I(X;Y|J)\end{align}
where 
\begin{align}
	J=\begin{cases}
0 & \text{ if } Z=(0,0),\\
1&\text{ if }Z=(1,1),\\
\mathtt{e}&\text{otherwise}
\end{cases}
\end{align}
which takes on non-zero values for $\epsilon>\epsilon_2=\frac{p}{1-p}$, as in Theorem~\ref{theoremeq:eps12}. In Fig.~\ref{fig:SK06}, $S_{\text{ow}}(X;Y\|Z)$ denotes the one-way communication capacity. This curve is calculated as follows: for every fixed value of $\epsilon \in [0,1]$, $S_{\text{ow}}(X;Y\|Z)$ is the maximum of
$I(U;Y|V)-I(U;Z|V)$
over all $p_{UV|X}$. The above expression is the upper concave envelope of the curve
\begin{align}
   p(X=0)\mapsto  \max_{p_{U|X}} I(U;Y)-I(U;Z)
\end{align}
at $p(X=0)=0.5$. Since the distribution $p_{YZ|X}$ is symmetric and $X$ is uniform in a DSBE source, using the symmetrization idea of \cite{nair2013upper} we obtain
\begin{align}
  S_{\text{ow}}(X;Y\|Z)=\max_{p_{UX}}I(U;Y)-I(U;Z). \label{nsym}
\end{align}
In fact, simulations indicate that the maximum in \eqref{nsym} is achieved when $X$ is uniform, indicating that auxiliary variable $V$ is not necessary to compute $S_{\text{ow}}(X;Y\|Z)$ for the DSBE source. The curve for $S_{\text{ow}}(X;Y\|Z)$ attains non-zero values for $\epsilon> 4p(1-p)$, which is the case also for $\bar{L}(X;Y\|Z)$ due to Theorem~\ref{thm:A5}. Similarly, we plot the rates achieved by the repetition codes of \cite{Maurer93} that are multi-letter and multi-round protocols. The $N$-repetition code achieves the SK rate
\begin{align}
	&R_N(X;Y\|Z) = \frac{p^N+{(1-p)}^N}{N}\max\Bigg\{0,\epsilon^N-h\Bigg(\frac{p^N}{p^N+{(1-p)}^N}\Bigg)\Bigg\}.
\end{align} 
In Fig.~\ref{fig:SK06}, we plot the rates for $N=2,3,4,5,6$. Fig.~\ref{fig:SK06} illustrates that there is a large gap between the lower bounds $R_N(X;Y\|Z)$ for $N=2,3,4,5,6$ and $B_0^{\text{sub}}(X;Y\|Z)$ for $\frac{p}{1-p}\leq \epsilon<1$.

\section{Proofs}\label{sec:Proofs}

\subsection{Proof of Theorem \ref{genertlsd}}\label{eq:genertlsd}
Without loss of generality, suppose the symbol $0$ is in both $\mathcal{X}$ and $\mathcal{Y}$. Suppose that $q_{XYZ}(x,y,z)=a(x)\,b(y)\,p_{XYZ}(x,y,z)$. Let $\bar{a}=\max_{x}a(x)>0$ and $\bar{b}=\max_{y}b(y)>0$. 

We define $X'$ and $Y'$ on $\mathcal{X}$ and $\mathcal{Y}$, respectively, as follows:
\begin{align}
	p_{XYZX'Y'} = p_{XYZ} \, p_{X'|X} \, p_{Y'|Y} 
\end{align}
where $p_{X'|X}(x'|x)$ and $p_{Y'|Y}(y'|y)$ satisfy
\begin{align}
\begin{array}{l}
  p_{X'|X}(0|x) = a(x) / \bar{a} \\
  p_{Y'|Y}(0|y) = b(y) / \bar{b}.
\end{array}
\label{genertlsdv1}
\end{align}
 The values of $p_{X'|X}(x'|x)$ and $p_{Y'|Y}(y'|y)$ 
for $x',y'\neq 0$ are not important for the proof. Observe that \eqref{genertlsdv1} implies 
\begin{align}
	p_{X'Y'}(0,0) & =\sum_{x,y,z}\frac{a(x)}{\bar{a}}\frac{b(y)}{\bar{b}} p_{XYZ}(x,y,z) = \frac{1}{\bar{a}\bar{b}}\sum_{x,y,z}q_{XYZ}(x,y,z)=\frac{1}{\bar{a}\bar{b}}>0
\end{align}
and
\begin{align}
	p_{XYZ|X'Y'}(x,y,z|0,0)
	& =a(x)b(y)p_{XYZ}(x,y,z) =q_{XYZ}(x,y,z).
\end{align}

Suppose that Alice, Bob and Eve observe  i.i.d.\ repetitions of $X,Y,Z$ according to $p_{XYZ}$. We now show that they can simulate i.i.d.\ repetitions according to $q_{XYZ}$. Alice has access to $X^n$. She passes $X^n$ through $\prod_{i=1}^np_{X_i'|X_i}$ to produce a sequence $X^{'n}$. Alice then puts into the public channel the list of indices $i$ such that $X'_i=0$. Similarly, Bob passes $Y^n$ through $\prod_{i=1}^np_{Y_i'|Y_i}$ to produce $Y^{'n}$ and puts into the public channel the list of indices $i$ such that $Y'_i=0$. Alice and Bob then consider the observations $(X_i, Y_i)$ for indices $i$ where $(X'_i,Y'_i)=(0,0)$, and discard their observations for other indices. The induced pmf on $(X_i, Y_i, Z_i)$ given the event $(X'_i,Y'_i)=(0,0)$ is $q_{XYZ}$. Alice and Bob can now proceed with any key agreement protocol for $q_{XYZ}$ that achieves a positive key rate.

\subsection{Proof of Theorem \ref{gentheoremeq:eps12} }\label{proofgentheoremeq:eps12} 
We prove the equivalence by showing that (ii) implies (i), (i) implies (iii), (iii) implies (iv), and (iv) implies (ii). The fact that (iii) implies (iv) is trivial, so we prove the other three implications in the following subsubsections. 

\subsubsection{  (ii) implies (i) }
We claim that proving Corollary \ref{corollary-new-1}  establishes the claim that  (ii) implies (i). To see this, assume that \eqref{eq:nproductnonzeroSKrate} holds for some  integer $n$ and disjoint non-empty sets $\mathcal{A}_1, \mathcal{A}_2\subset \mathcal{X}^n$, and disjoint non-empty sets $\mathcal{B}_1, \mathcal{B}_2\subset \mathcal{Y}^n$. Let $X'\in\{1,2,3\}$ be a function of $X^n$ defined as follows: $X'=1$ if $X^n\in\mathcal{A}_1$, $X'=2$ if $X^n\in\mathcal{A}_2$ and $X'=3$  otherwise. We defined $Y'$ as a function of $Y^n$ in a similar manner using $\mathcal{B}_1$ and $\mathcal{B}_2$. Finally, let $Z'=Z^n$. We have, $S(X;Y\|Z)>0$ if $S(X';Y'\|Z')>0$ since Alice and Bob can produce $X'$ and $Y'$ from $X^n$ and $Y^n$ respectively. Equation \eqref{eqn:gentheoremeqA2} for $(X', Y', Z')$ with the choice $x'_1=1, x'_2=2, y'_1=1, y'_2=2$ is equivalent to \eqref{eq:nproductnonzeroSKrate} for the triple $(X^n, Y^n, Z^n)$ with the sets $\mathcal{A}_1, \mathcal{A}_2, \mathcal{B}_1$ and $\mathcal{B}_2$. 

It remains to prove Corollary \ref{corollary-new-1}. In other words, we wish to prove that $S(X;Y\|Z)>0$ if (\ref{eqn:gentheoremeqA2}) holds for distinct symbols $x_1, x_2 \in\mathcal{X}$ and distinct symbols $y_1, y_2 \in\mathcal{Y}$. Let $p_{ij} = p_{XY}(x_i, y_{j})$ for $i,j=1,2$. By \eqref{eqn:gentheoremeqA2}, we have
\begin{align}
  \frac12 \log\left(\frac{p_{11}p_{22}}{p_{12}p_{21}}\right)>0
\end{align}
or equivalently
\begin{align}
p_{11}p_{22} > p_{12}p_{21}.\label{eqdfrtgdf4545}
\end{align} 
Consider some even natural number $n$ and the sets
\begin{align}
	\mathcal{A}=\{\mathbf{x}_1, \mathbf{x}_2\},\quad \mathcal{B}=\{\mathbf{y}_1, \mathbf{y}_2\}
\end{align}
where
\begin{align}
	&\mathbf{x}_1=(x_1,x_1,\cdots, x_1, x_2,x_2,\cdots, x_2), \nonumber\\
	&\mathbf{x}_2= (x_2,x_2,\cdots, x_2,x_1,x_1,\cdots, x_1 ),\nonumber\\
	&\mathbf{y}_1=(y_1,y_1,\cdots, y_1,\, y_2,\, y_2,\cdots,\, y_2),\nonumber\\
	&\mathbf{y}_2= (\underbrace{y_2,y_2,\cdots, y_2}_{n/2},\,\underbrace{y_1,\, y_1,\cdots,\, y_1}_{n/2} ).
\end{align}

As in the proof of Theorem \ref{OrlitskyWigdersonThm} and Maurer's example in \cite[p.~740]{Maurer93}, suppose that Alice and Bob observe $N$ independent blocks, each having i.i.d.\ realizations of $(X,Y)$. For each block, Alice declares whether  $X^n\in\mathcal{A}$ and Bob declares whether $Y^n\in\mathcal{B}$. If $X^n\notin\mathcal{A}$ or $Y^n\notin\mathcal{B}$, they discard the block. Otherwise, they keep the block and use it for SK agreement. To prove that key generation is feasible, it suffices to show that 
\begin{align}
&I(X^n;Y^n|X^n\in\mathcal{A}, Y^n\in \mathcal{B})> I(X^n;Z^n|X^n\in\mathcal{A}, Y^n\in \mathcal{B}) \label{fkjbeg34345}
\end{align}
for large $n$. Equation (\ref{fkjbeg34345}) implies that  the legitimate users have a simple entropic advantage over the eavesdropper (see Definition \ref{def:advantage}) and hence a positive key rate can be achieved. We now show that (\ref{fkjbeg34345}) is satisfied. For any three random variables $X,Y,Z$ we have
\begin{align}
&I(X;Y)-I(X;Z)= H(X,Y|Z)-H(Y|X,Z)-H(X|Y)\geq H(X,Y|Z)-H(Y|X)-H(X|Y).
\end{align}
Thus, it suffices to show that for large $n$ we have
\begin{align}
&H(X^n,Y^n|Z^n, X^n\in\mathcal{A}, Y^n\in \mathcal{B})>H(Y^n|X^n, X^n\in\mathcal{A}, Y^n\in \mathcal{B})+H(X^n|Y^n, X^n\in\mathcal{A}, Y^n\in \mathcal{B}).\label{eqjstrtg}
\end{align}

We compute
\begin{align}
	&\mathbb{P}[X^n=\mathbf{x}_1, Y^n=\mathbf{y}_1]=\mathbb{P}[X^n=\mathbf{x}_2, Y^n=\mathbf{y}_2]=p_{11}^{n/2}p_{22}^{n/2} \\
	&\mathbb{P}[X^n=\mathbf{x}_2, Y^n=\mathbf{y}_1]=\mathbb{P}[X^n=\mathbf{x}_1, Y^n=\mathbf{y}_2]=p_{12}^{n/2}p_{21}^{n/2}.
\end{align}
The conditional pmf of $(X^n, Y^n)$ given that $X^n\in\mathcal{A}$ and $Y^n\in \mathcal{B}$ is 
\begin{align}
\mathbb{P}[X^n=\mathbf{x}_1, Y^n=\mathbf{y}_1|X^n\in\mathcal{A}, Y^n\in \mathcal{B}]\nonumber&=\mathbb{P}[X^n=\mathbf{x}_2, Y^n=\mathbf{y}_2|X^n\in\mathcal{A}, Y^n\in \mathcal{B}]
\nonumber\\
&=\frac{p_{11}^{n/2}p_{22}^{n/2}}{2(p_{11}^{n/2}p_{22}^{n/2}+p_{12}^{n/2}p_{21}^{n/2})}\label{eq:symmetricforposSK1}
\end{align}
\begin{align}
\mathbb{P}[X^n=\mathbf{x}_2, Y^n=\mathbf{y}_1|X^n\in\mathcal{A}, Y^n\in \mathcal{B}]
&=\mathbb{P}[X^n=\mathbf{x}_1, Y^n=\mathbf{y}_2|X^n\in\mathcal{A}, Y^n\in \mathcal{B}]
\nonumber\\
&=\frac{p_{21}^{n/2}p_{12}^{n/2}}{2(p_{11}^{n/2}p_{22}^{n/2}+p_{12}^{n/2}p_{21}^{n/2})}.\label{eq:symmetricforposSK2}
\end{align}
If $X^n\in\mathcal{A}$ and $Y^n\in \mathcal{B}$, then we can model the conditional joint pmf of $(X^n,Y^n)$ as a DSBS with parameter
\begin{align}
  \tilde{p}_n=\frac{p_{11}^{n/2}p_{22}^{n/2}}{p_{11}^{n/2}p_{22}^{n/2}+p_{12}^{n/2}p_{21}^{n/2}}
\end{align}
due to symmetry in (\ref{eq:symmetricforposSK1}) and (\ref{eq:symmetricforposSK2}). Thus, we obtain
\begin{align}
	&H(X^n|Y^n, X^n\in\mathcal{A}, Y^n\in \mathcal{B})=H(Y^n|X^n, X^n\in\mathcal{A}, Y^n\in \mathcal{B})=h(\tilde{p}_n)
\end{align}
where $h(\cdot)$ is the binary entropy function. We have $h(p)\leq -2(1-p)\log(1-p)$ for any $p\in[0.5,1]$. Using \eqref{eqdfrtgdf4545}, we have $\tilde{p}_n\in [0.5,1]$ and 
\begin{align}
  \lim_{n\rightarrow\infty}(1-\tilde{p}_n)^{\frac{1}{n}}= \left(\frac{
  p_{12}p_{21}}{p_{11}p_{22}}\right)^{\frac 12}.
\end{align}
Hence, we have
\begin{align}
	&\lim_{n\rightarrow\infty} h(\tilde{p}_n)^{\frac{1}{n}}\leq \lim_{n\rightarrow\infty}\left(-2(1-\tilde{p}_n)\log(1-\tilde{p}_n)\right)^{\frac{1}{n}}\nonumber=\left(\frac{
	p_{12}p_{21}}{p_{11}p_{22}}\right)^{\frac 12}
\end{align}
and we obtain
\begin{align}
\lim_{n\rightarrow\infty}H(X^n|Y^n, X^n\in\mathcal{A}, Y^n\in \mathcal{B})^{\frac{1}{n}}&=\lim_{n\rightarrow\infty}H(Y^n|X^n, X^n\in\mathcal{A}, Y^n\in \mathcal{B})^{\frac{1}{n}}\leq 
\left(\frac{
 p_{12}p_{21}}{p_{11}p_{22}}\right)^{\frac 12}\label{eqdf23erNN}.
\end{align}
This equation implies that
\begin{align}
&\lim_{n\rightarrow\infty}\Big[H(X^n|Y^n, X^n\in\mathcal{A}, Y^n\in \mathcal{B})+H(Y^n|X^n, X^n\in\mathcal{A}, Y^n\in \mathcal{B})\Big]^{\frac{1}{n}}\leq 
\left(\frac{
 p_{12}p_{21}}{p_{11}p_{22}}\right)^{\frac 12}\label{eqdf23er}
\end{align}
which gives a bound on the asymptotics of the right hand side in \eqref{eqjstrtg}. We now consider the term on the left hand side in (\ref{eqjstrtg}). Our aim is to show that
\begin{align}
&\liminf_{n\rightarrow\infty}H(X^n, Y^n|Z^n, X^n\in\mathcal{A}, Y^n\in \mathcal{B})^{\frac{1}{n}}\geq \exp\Big(\!-\frac12D_{\frac 12}\big(p_{Z|XY}(\cdot |x_1,y_1)\big\| p_{Z|XY}(\cdot |x_{2},y_{2})\big)\Big).
\label{eqn34df45gd}
\end{align}
This equation together with \eqref{eqdf23er} show that \eqref{eqjstrtg} holds for large values of $n$ if
\begin{align}
	&\left(\frac{p_{12}p_{21}}{p_{11}p_{22}}\right)^{\frac 12} <\exp\Big(-\frac12D_{\frac 12}\big(p_{Z|XY}(\cdot |x_1,y_1)\big\| p_{Z|XY}(\cdot |x_{2},y_{2})\big)\Big)
\end{align}
which is equivalent to the condition
\begin{align}
	&D_{\frac 12}\big(p_{Z|XY}(\cdot |x_1,y_1)\big\| p_{Z|XY}(\cdot |x_{2},y_{2})\big)< \log\left(\frac{p_{11}p_{22}}{p_{12}p_{21}}\right).
\end{align}

It remains to prove \eqref{eqn34df45gd}.  From the perspective of Eve who observes $Z^n$, there are four possibilities of $(X^n,Y^n)=(\mathbf{x}_i, \mathbf{y}_j)$ for $i,j\in\{1,2\}$. Eve can view this as a hypothesis testing problem. For example, given $(X^n,Y^n)=(\mathbf{x}_1, \mathbf{y}_1)$, the conditional pmf of $(Z_i, Z_{\frac{n}{2}+i})$ satisfies
\begin{align}
  p_{Z|XY}(z_i|x_1,y_1)\cdot p_{Z|XY}(z_{\frac{n}{2}+i}|x_2,y_2)
\end{align}
for all $1\leq i\leq \frac{n}{2}-1$. Furthermore, $Z_i$ and $Z_{\frac{n}{2}+i}$ are conditionally independent given $(X^n,Y^n)=(\mathbf{x}_1, \mathbf{y}_1)$ for all $1\leq i\leq \frac{n}{2}-1$. Therefore, given the hypothesis $(X^n,Y^n)=(\mathbf{x}_1, \mathbf{y}_1)$, Eve observes $n/2$ i.i.d.\ repetitions
\begin{align}
  q^{(11)}_{Z_a Z_b}(z_a, z_b)= p_{Z|XY}(z_a|x_1,y_1)p_{Z|XY}(z_{b}|x_2,y_2).
\end{align}
More generally, given the hypothesis $(X^n,Y^n)=(\mathbf{x}_i, \mathbf{y}_j)$, Eve observes $n/2$ i.i.d.\ repetitions
\begin{align}
  q^{(ij)}_{Z_a Z_b}(z_a, z_b)=p_{Z|XY}(z_a|x_i,y_j)p_{Z|XY}(z_{b}|x_{3-i},y_{3-j})
\end{align}
for $i,j\in\{1,2\}$. 

We remark that the prior probability of the hypothesis $(X^n,Y^n)=(\mathbf{x}_i, \mathbf{y}_j)$ depends on $n$; see \eqref{eq:symmetricforposSK1} and \eqref{eq:symmetricforposSK2}. Therefore, we cannot directly apply results from the hypothesis testing literature, where fixed prior hypothesis probabilities are assumed. We use the following lemma.
\begin{lemma}\label{lem:UVlemma}\cite[Eq. (10)]{kanaya1995asymptotics}
For any $p_{UV}$ and any two distinct symbols $u_1, u_2$, we have
\begin{align}
   &\frac{H(U|V)}{ \log(2)}\geq\Bigg( p_{U}(u_1)\sum_{v\in\mathcal{D}}p_{V|U}(v|u_1) +p_U(u_2)\sum_{v\in\mathcal{D}^c}p_{V|U}(v|u_2) \Bigg) 
\end{align}
where $\mathcal{D}=\{v:~p_{UV}(u_1,v)<p_{UV}(u_2,v)\}$.
\end{lemma}

We apply Lemma~\ref{lem:UVlemma} with $U=(X^n, Y^n)$, $V=Z^n$, and
\begin{align}
&p_{UV}((x^n, y^n), z^n)=p_{X^n, Y^n, Z^n}(x^n, y^n, z^n|X^n\in\mathcal{A}, Y^n\in \mathcal{B}) \\
&u_1=(\mathbf{x}_1,\mathbf{y}_1) \\
&u_2=(\mathbf{x}_2,\mathbf{y}_2).
\end{align}
Using \eqref{eq:symmetricforposSK1}, we have
\begin{align}
&p_U(u_1)=p_U(u_2)=\frac{p_{11}^{n/2}p_{22}^{n/2}}{2(p_{11}^{n/2}p_{22}^{n/2}+p_{12}^{n/2}p_{21}^{n/2})}
\end{align}
and we obtain
\begin{align}
&H(X^n,Y^n|Z^n, X^n\in\mathcal{A}, Y^n\in \mathcal{B})^{\frac{1}{n}}\nonumber\\
&\geq  
\left(\log(2)\frac{p_{11}^{n/2}p_{22}^{n/2}}{p_{11}^{n/2}p_{22}^{n/2}+p_{12}^{n/2}p_{21}^{n/2}}\right)^{\frac{1}{n}}\nonumber\\
&\;
\times\!\Bigg\{\frac12 \!\sum_{z^n\in\mathcal{D}}p_{Z^n|X^n, Y^n}\!(z^n|\mathbf{x}_1, \mathbf{y}_1, X^n\!\in\!\mathcal{A}, Y^n\!\in\! \mathcal{B})  +\!\frac12\! \sum_{z^n\in~\mathcal{D}^c}\!p_{Z^n|X^n, Y^n}(z^n|\mathbf{x}_2, \mathbf{y}_2, X^n\!\in\!\mathcal{A}, Y^n\!\in\! \mathcal{B}) \Bigg\}^{\frac{1}{n}}.\label{eqn34fef45ewf}
\end{align} 
Using \eqref{eqdfrtgdf4545}, we have \begin{align}\lim_{n\rightarrow\infty}\left(\log(2)\frac{p_{11}^{n/2}p_{22}^{n/2}}{p_{11}^{n/2}p_{22}^{n/2}+p_{12}^{n/2}p_{21}^{n/2}}\right)^{\frac{1}{n}}=1.\label{fdrkjebef1}\end{align}

Next, observe that 
\begin{align}
	&\mathcal{D}=\Big\{z^n: p_{Z^n|X^n, Y^n}(z^n|\mathbf{x}_1, \mathbf{y}_1)<p_{Z^n|X^n, Y^n}(z^n|\mathbf{x}_2, \mathbf{y}_2)\Big\}
\end{align}
is the maximum a-posteriori probability (MAP) decision region for a new binary hypothesis testing problem with two equiprobable hypotheses, i.e., $(X^n, Y^n)=(\mathbf{x}_1, \mathbf{y}_1)$ and $(X^n, Y^n)=(\mathbf{x}_2, \mathbf{y}_2)$. In this problem, 
Eve observes $n/2$ i.i.d.\ repetitions of $q^{(11)}_{Z_a Z_b}$ under the first hypothesis, and $n/2$ i.i.d.\ repetitions of $q^{(22)}_{Z_a Z_b}$ under the second hypothesis. The expression
\begin{align}
&\frac12 \sum_{z^n\in\mathcal{D}}p_{Z^n|X^n, Y^n}(z^n|\mathbf{x}_1, \mathbf{y}_1, X^n\in\mathcal{A}, Y^n\in \mathcal{B}) +\frac12 \sum_{z^n\in\mathcal{D}^c}p_{Z^n|X^n, Y^n}(z^n|\mathbf{x}_2, \mathbf{y}_2, X^n\in\mathcal{A}, Y^n\in \mathcal{B}) 
\end{align}
is the error probability, which is asymptotically equal to $\exp(-\frac{n}{2}E)$, where
\begin{align}
E&=C\Big(q^{(11)}_{Z_a Z_b}\big\| q^{(22)}_{Z_a Z_b}\Big)=D_{\frac 12}\Big(p_{Z|XY}(\cdot |x_1,y_1)\big\| p_{Z|XY}(\cdot |x_2,y_2)\Big).
\label{eqnEval}
\end{align}
Equation \eqref{eqnEval} follows from the argument given in Example \ref{RnyiHalf}. 
Therefore, we obtain
\begin{align}
&\lim_{n\rightarrow\infty}\!\Bigg\{\frac12 \sum_{z^n\!\in\!\mathcal{D}}p_{Z^n|X^n, Y^n}(z^n|\mathbf{x}_1, \mathbf{y}_1, X^n\!\in\!\mathcal{A}, Y^n\in \mathcal{B}) +\frac12 \sum_{z^n\!\in\!\mathcal{D}^c}p_{Z^n|X^n, Y^n}(z^n|\mathbf{x}_2, \mathbf{y}_2, X^n\!\in\!\mathcal{A}, Y^n\!\in\! \mathcal{B}) \Bigg\}^{\frac{1}{n}}\nonumber
\\&=\exp\Big(-\frac12D_{\frac 12}\big(p_{Z|XY}(\cdot |x_1,y_1)\| p_{Z|XY}(\cdot |x_2,y_2)\big)\Big).\label{fdkjbretdv435}
\end{align}
Combining \eqref{eqn34fef45ewf}, \eqref{fdrkjebef1}, and \eqref{fdkjbretdv435}  establishes \eqref{eqn34df45gd}.

\subsubsection{   (i) implies (iii) }

Suppose that $S(X;Y\|Z)>0$. Because Alice and Bob can produce a key at a positive rate, they can also produce a key of length one bit. Maurer and Wolf \cite{maurer2000information} show the equivalence of the strong and weak notions of security for the source model problem. More specifically, from \cite{maurer2000information} and using $S(X;Y\|Z)>0$, we conclude that, given any $\delta>0$, there is an interactive communication protocol yielding  bits $K_A$ and $K_B$ for Alice and Bob such that $\mathbb{P}[K=K_A=K_B]\geq 1-\delta$ for some uniform bit $K\in\{1,2\}$. Furthermore, we have
\begin{align}
   I(K;Z^n\mathbf{F})\leq \delta.
\end{align}

The triangle inequality gives
\begin{align}\nonumber
\|p_{K_A K_B Z^n \mathbf{F}} - q_{K_AK_B}\cdot p_{Z^n \mathbf{F}}  \|_{TV}&\leq
\|p_{K_A K_B Z^n \mathbf{F}} - p_{K_AK_B}\cdot p_{Z^n \mathbf{F}}  \|_{TV}\!+\!
\|p_{K_AK_B}\cdot p_{Z^n \mathbf{F}}  - q_{K_AK_B}\cdot p_{Z^n \mathbf{F}} \|_{TV}\nonumber
\\&=
\|p_{K_A K_B Z^n \mathbf{F}} - p_{K_AK_B}\cdot p_{Z^n \mathbf{F}}  \|_{TV}+
\|p_{K_A K_B} - q_{K_AK_B}  \|_{TV}\label{j34bgkuret}
\end{align}
and by Fano's inequality we have
\begin{align}
	&I(KK_AK_B;Z^n\mathbf{F})\leq I(K;Z^n\mathbf{F})+H(K_AK_B|K)\leq \delta+h(\delta)+3\delta.
\end{align}
Therefore, we have the bound
\begin{align}
	&D(p_{K_AK_BZ^n\mathbf{F}}\|p_{K_AK_B}p_{Z^n\mathbf{F}})=I(K_AK_B;Z^n\mathbf{F})\leq 4\delta+h(\delta).
\end{align}
By Pinsker's inequality, we have
\begin{align}
	&\|p_{K_AK_BZ^n\mathbf{F}}-p_{K_AK_B}p_{Z^n\mathbf{F}}\|_{TV}\nonumber\leq \sqrt{2\delta+\frac 12 h(\delta)}\label{lkjbrtjbdf342}.
\end{align}
Next, from uniformity of $K$ and 
$\mathbb{P}[K=K_A=K_B]\geq 1-\delta$, we have
\begin{align}
	\left|p_{K_AK_B}(i,i)-\frac12\right|\leq \delta\; \text{  for } i=1,2
\end{align}
and
\begin{align}
	p_{K_AK_B}(i,j)\leq \delta\; \text{  for } i\neq j.
\end{align}
Therefore, we can write
\begin{align}
\|p_{K_A K_B} - q_{K_AK_B}  \|_{TV}\leq 2\delta.\label{lkjbrtjbdf34}
\end{align}
From \eqref{j34bgkuret}, \eqref{lkjbrtjbdf342} and \eqref{lkjbrtjbdf34}, we obtain
\begin{align}
&\|p_{K_A K_B Z^n \mathbf{F}} - q_{K_AK_B}\cdot p_{Z^n \mathbf{F}}  \|_{TV}\leq
\sqrt{2\delta+\frac 12 h(\delta)}+2\delta.
\end{align}
The right hand side of the above equation tends to zero as $\delta$ tends to zero.

\subsubsection{ (iv) implies (ii) }

It suffices to prove Lemma  \ref{lemma0-bnwe} below. This lemma  identifies sets $\mathcal{A}_i$ and $\mathcal{B}_i$ such that (\ref{eq:nproductnonzeroSKrate}) holds if 
\begin{align}
-\log\left(1-4\delta\right)<\log\frac{\frac 12 -2\delta}{2\delta}.\label{kjbsdf,mbr234343}
\end{align}
Equation \eqref{kjbsdf,mbr234343} holds for $\delta< \frac{3-\sqrt{5}}{8}$. This completes the proof.


\begin{lemma}\label{lemma0-bnwe} Consider a code with source sequences of length $n$,  interactive communication $\mathbf{F}=(F_1, F_2, ...)$ satisfying \eqref{eqn-=deertg}, and secret key bits $K_A\in\{1,2\}$ and $K_B\in\{1,2\}$, created by Alice and Bob after public discussion. We define
\begin{align}
	\delta=\|p_{K_A K_B Z^n \mathbf{F}} - q_{K_AK_B}\cdot p_{Z^n, \mathbf{F}}  \|_{TV}.
\end{align}
One can find disjoint non-empty subsets $\mathcal{A}_1, \mathcal{A}_2\subset\mathcal{X}^n$ and disjoint non-empty subsets $\mathcal{B}_1, \mathcal{B}_2\subset \mathcal{Y}^n$ such that 
\begin{align}
	&\frac12D_{\frac 12}\Big(p_{Z^n}(\cdot |X^n\!\in\!\mathcal{A}_1,Y^n\!\in\!\mathcal{B}_1)\big\|  p_{Z^n}(\cdot |X^n\!\in\!\mathcal{A}_2,Y^n\!\in\!\mathcal{B}_2)\Big)\leq -\log\left(1-4\delta\right)\label{eq:IamOutofNamesforEquations----}
\end{align}	
and
\begin{align}
&\frac12
	\log\left(\frac{
	\mathbb{P}[X^n\!\in\!\mathcal{A}_1,Y^n\!\in\!\mathcal{B}_1]\mathbb{P}[X^n\!\in\!\mathcal{A}_2,Y^n\!\in\!\mathcal{B}_2]
}{
\mathbb{P}[X^n\!\in\!\mathcal{A}_1,Y^n\!\in\!\mathcal{B}_2]\mathbb{P}[X^n\in\mathcal{A}_2,Y^n\!\in\!\mathcal{B}_1]
}\right)\geq 
	\log\frac{\frac 12 -2\delta}{2\delta}.\label{eqn:lkjbkzmsdjbf----}
\end{align}
\end{lemma}
\begin{proof}
There is a realization $\mathbf{F}=\mathbf{f}$ such that $\mathbb{P}[\mathbf{F}=\mathbf{f}]>0$ and
\begin{align}
	\Big\|p_{K_A K_B Z^n|\mathbf{F}}(\cdot |\mathbf{f})
	  -  q_{K_AK_B} \cdot p_{Z^n|\mathbf{F}}(\cdot |\mathbf{f}) \Big\|_{TV}\leq \delta.\label{eqn:slfjbkjbt}
\end{align}
From the data processing property of the total variation distance, we have
\begin{align}
	\big\|p_{K_A K_B|\mathbf{F}}(\cdot |\mathbf{f}) \!-\! q_{K_AK_B} \big\|_{TV}
	\le \delta\label{eqnrefklh23}
\end{align}
and therefore
\begin{align}\mathbb{P}[K_A=1, K_B=1|\mathbf{F}=\mathbf{f}]&\geq \frac 12 -2\delta\label{eqnrefklh231}\\
\mathbb{P}[K_A=2, K_B=2|\mathbf{F}=\mathbf{f}]&\geq \frac 12 -2\delta\label{eqnrefklh232}\\
\mathbb{P}[K_A=1, K_B=2|\mathbf{F}=\mathbf{f}]&\leq 2\delta \label{eqnrefklh233}\\
\mathbb{P}[K_A=2, K_B=1|\mathbf{F}=\mathbf{f}]&\leq 2\delta. 
\label{eqnrefklh234}
\end{align}

As in the proof of Theorem \ref{OrlitskyWigdersonThm},  the conditional pmfs of $(X^n, Y^n, Z^n)$ given $\mathbf{F}=\mathbf{f}$ have the form $p_r(x^n, y^n, z^n)$ given in \eqref{eqn:AAAd13} for some sets $\mathcal{A}$ and $\mathcal{B}$ that depend on $\mathbf{f}$. Furthermore, given $\mathbf{F}=\mathbf{f}$, the key $K_A$ is a function of $X^n$. We partition $\mathcal{A}$ into  $\mathcal{A}_1\cup \mathcal{A}_2$ as follows:
 $\mathcal{A}_i=\{x^n: K_A(x^n, \mathbf{f})=i\}$ for $i=1,2$. We define $\mathcal{B}_i$ similarly using $K_B$. 

Observe for $i,j\in\{1,2\}$ that
\begin{align}
	&\mathbb{P}[K_A=i, K_B=j|\mathbf{F}=\mathbf{f}]=\frac{\mathbb{P}[X^n\in\mathcal{A}_i,Y^n\in\mathcal{B}_j]}{ \mathbb{P}[X^n\in\mathcal{A},Y^n\in\mathcal{B}]}.\label{lasnkljeb}
\end{align}
From \eqref{eqnrefklh231}-\eqref{lasnkljeb}, we obtain 
\begin{align}
&\frac12
	\log\left(\frac{
	\mathbb{P}[X^n\in\mathcal{A}_1,Y^n\in\mathcal{B}_1]\mathbb{P}[X^n\in\mathcal{A}_2,Y^n\in\mathcal{B}_2]
}{
\mathbb{P}[X^n\in\mathcal{A}_1,Y^n\in\mathcal{B}_2]\mathbb{P}[X^n\in\mathcal{A}_2,Y^n\in\mathcal{B}_1]
}\right)\nonumber\\
&\quad=\frac12
	\log\Bigg(
	\frac{
	\mathbb{P}[K_A\!=\!1, K_B\!=\!1|\mathbf{F}\!=\!\mathbf{f}]
}{
\mathbb{P}[K_A\!=\!1, K_B\!=\!2|\mathbf{F}\!=\!\mathbf{f}]
}\times
\frac{
	\mathbb{P}[K_A\!=\!2, K_B\!=\!2|\mathbf{F}\!=\!\mathbf{f}]
}{
\mathbb{P}[K_A\!=\!2, K_B\!=\!1|\mathbf{F}\!=\!\mathbf{f}]
}\Bigg)\nonumber\\
&\quad\geq \frac12
	\log\frac{(\frac 12 -2\delta)^2}{(2\delta)^2}= 
	\log\frac{\frac 12 -2\delta}{2\delta}.\label{eqn:lkjbkzmsdjbf}
\end{align}

Next, from \eqref{eqn:slfjbkjbt} we have \eqref{eqn:slfjbkjbt23} given at the top of this page,
\begin{figure*}
\begin{align}
	&\sum_{z^n} \Bigg|\mathbb{P}[K_A\!=\!1, K_B\!=\!1|\mathbf{F}\!=\!\mathbf{f}] \times p_{Z^n|K_A, K_B,\mathbf{F}}(z^n|k_A\!=\!1, k_B\!=\!1,\mathbf{f})-
\frac 12 p_{Z^n|\mathbf{F}}(z^n|\mathbf{f}) 
\Bigg|\nonumber\\
&+
	\sum_{z^n}\! \Bigg|\mathbb{P}[K_A\!=\!2, K_B\!=\!2|\mathbf{F}\!=\!\mathbf{f}] \times p_{Z^n|K_A, K_B,\mathbf{F}}(z^n|k_A\!=\!2, k_B\!=\!2,\mathbf{f}) - \frac 12 p_{Z^n|\mathbf{F}}(z^n|\mathbf{f}) \Bigg|
\nonumber\\[0.5ex]
&\quad\leq 2\delta.\label{eqn:slfjbkjbt23}
\\
\hline \nonumber
\\ \nonumber
	&\sum_{z^n} \Bigg|\mathbb{P}[K_A\!=\!1, K_B\!=\!1|\mathbf{F}\!=\!\mathbf{f}] \times p_{Z^n|K_A, K_B\mathbf{F}}(z^n|k_A\!=\!1, k_B\!=\!1,\mathbf{f}) \nonumber
\\&\qquad - \mathbb{P}[K_A\!=\!2, K_B\!=\!2|\mathbf{F}\!=\!\mathbf{f}] \times p_{Z^n|K_A, K_B\mathbf{F}}(z^n|k_A\!=\!2, k_B\!=\!2,\mathbf{f})  \Bigg|\nonumber\\[0.5ex]
&\quad
\leq 2\delta.\label{eqn:slfjbkjbt2323}
\\
\hline\nonumber
\end{align}
\end{figure*}
and the triangle inequality gives \eqref{eqn:slfjbkjbt2323} also given at the top of this page. The triangle inequality further implies equation \eqref{eqjklbgkb} given at the top of the next page, where we use \eqref{eqnrefklh23} and \eqref{eqn:slfjbkjbt2323} in the last step of the derivation. 
\begin{figure*}
\begin{align}\nonumber
&\sum_{z^n} \Bigg|\frac 12 p_{Z^n|K_A, K_B\mathbf{F}}(z^n|k_A=1, k_B=1,\mathbf{f})- \frac 12 p_{Z^n|K_A, K_B\mathbf{F}}(z^n|k_A=2, k_B=2,\mathbf{f})  \Bigg|\nonumber 
\\&
\leq\sum_{z^n} \Bigg|\mathbb{P}[K_A=1, K_B=1|\mathbf{F}=\mathbf{f}]
\times p_{Z^n|K_A, K_B\mathbf{F}}(z^n|k_A=1, k_B=1,\mathbf{f}) \nonumber
\\&\qquad\quad - \mathbb{P}[K_A=2, K_B=2|\mathbf{F}=\mathbf{f}]
\times p_{Z^n|K_A, K_B\mathbf{F}}(z^n|k_A=2, k_B=2,\mathbf{f})  \Bigg|\nonumber
\\&\quad+
\sum_{z^n} \Bigg|\mathbb{P}[K_A=1, K_B=1|\mathbf{F}=\mathbf{f}] \times p_{Z^n|K_A, K_B\mathbf{F}}(z^n|k_A=1, k_B=1,\mathbf{f})  -\frac 12 p_{Z^n|K_A, K_B\mathbf{F}}(z^n|k_A=1, k_B=1,\mathbf{f})\Bigg|\nonumber
\\&\quad+
\sum_{z^n} \Bigg|\mathbb{P}[K_A=2, K_B=2|\mathbf{F}=\mathbf{f}] \times p_{Z^n|K_A, K_B\mathbf{F}}(z^n|k_A=2, k_B=2,\mathbf{f}) -\frac 12 p_{Z^n|K_A, K_B\mathbf{F}}(z^n|k_A=2, k_B=2,\mathbf{f})\Bigg|\nonumber\\
&=\sum_{z^n} \Bigg|\mathbb{P}[K_A=1, K_B=1|\mathbf{F}=\mathbf{f}] \times p_{Z^n|K_A, K_B\mathbf{F}}(z^n|k_A=1, k_B=1,\mathbf{f}) \nonumber
\\&\qquad\quad - \mathbb{P}[K_A=2, K_B=2|\mathbf{F}=\mathbf{f}] \times p_{Z^n|K_A, K_B\mathbf{F}}(z^n|k_A=2, k_B=2,\mathbf{f})  \Bigg|\nonumber
\\&\quad+
 \left|\mathbb{P}[K_A=1, K_B=1|\mathbf{F}=\mathbf{f}] -\frac 12 \right|+
\left|\mathbb{P}[K_A=2, K_B=2|\mathbf{F}=\mathbf{f}] -\frac 12\right|\nonumber\\
&\leq \sum_{z^n} \Bigg|\mathbb{P}[K_A=1, K_B=1|\mathbf{F}=\mathbf{f}] \times p_{Z^n|K_A, K_B\mathbf{F}}(z^n|k_A=1, k_B=1,\mathbf{f}) \nonumber
\\[0.5ex]&\qquad\quad - \mathbb{P}[K_A=2, K_B=2|\mathbf{F}=\mathbf{f}] \times p_{Z^n|K_A, K_B\mathbf{F}}(z^n|k_A=2, k_B=2,\mathbf{f})  \Bigg|\nonumber
\\&\quad+
2\|p_{K_A K_B|\mathbf{F}}(k_A, k_B|\mathbf{f}) - q_{K_AK_B}(k_A, k_B) \|_{TV}\nonumber
\\[0.7ex]&\leq 2\delta+2\delta\label{eqjklbgkb}
\\
\hline \nonumber
\end{align}
\end{figure*}
Observe for $i\in\{1,2\}$ that
\begin{align}
	&p_{Z^n|K_AK_B\mathbf{F}}(z^n|k_A=i,k_B=i,\mathbf{f}) =p_{Z^n}(z^n|X^n\in\mathcal{A}_{i},Y^n\in\mathcal{B}_{i}).
\end{align}
Therefore, \eqref{eqjklbgkb} shows that 
\begin{align}
	&\Big\|p_{Z^n}(\cdot |X^n\in\mathcal{A}_{1},Y^n\!\in\!\mathcal{B}_{1})-p_{Z^n}(\cdot |X^n\!\in\!\mathcal{A}_{2},Y^n\!\in\!\mathcal{B}_{2})\Big\|_{TV}\leq 4\delta
\end{align}
and by \eqref{TV-Chernoff} we have
\begin{align}
	&\frac 12 D_{\frac 12}\Big(p_{Z^n}(\cdot |X^n\!\in\!\mathcal{A}_1,Y^n\!\in\!\mathcal{B}_1)\Big\|  p_{Z^n}(\cdot |X^n\!\in\!\mathcal{A}_2,Y^n\!\in\!\mathcal{B}_2)\Big)\leq -\log\left(1-4\delta\right).\label{eq:IamOutofNamesforEquations}
\end{align}	
\end{proof}

\subsection{Proof of Theorem \ref{mmthm1s}}\label{eq:maintheoremsproof}
\label{proof-lower-bound1}

In light of Examples \ref{example2n} and \ref{example3n}, we only need to to prove the first part of the theorem, namely $S(X;Y\|Z)=0$ if $\epsilon\leq \epsilon_1$, and 
$S(X;Y\|Z)>0$ if $\epsilon> \epsilon_2$. The fact that $S(X;Y\|Z)=0$ if $\epsilon\leq \epsilon_1$ follows from Theorem \ref{theoremWolf} if we can show that $\epsilon_1=1-d_{\emph{ind}}=\frac{1}{F(p_{XY})}$.
Observe that  $F(p_{XY})$ is (see Definition~\ref{def:Fpxy}) the minimum over all product measures $q_{XY}=q_X\, q_Y$ of 
\begin{align}
	\max_{x,y}\left(\frac{p_{XY}(x,y)}{q_{XY}(x,y)}\right)\cdot\max_{x,y}\left(\frac{q_{XY}(x,y)}{p_{XY}(x,y)}\right).\label{eq:Icouldnotfindagoodname}
\end{align}
This is because (\ref{eq:Icouldnotfindagoodname}) would not change if we multiply $q_{XY}(x,y)$ by a positive constant. Moreover, by the same argument, we can restrict attention to product measures $q_{XY}=q_X\, q_Y$ satisfying 
\begin{align}
	\max_{x,y}\left(\frac{q_{XY}(x,y)}{p_{XY}(x,y)}\right)=1.
\end{align}
The equality $\epsilon_1=\frac{1}{F(p_{XY})}$ follows because of the following alternative characterization of $\epsilon_1$.
\begin{theorem}\label{lemma-generaldis1}
Define a matrix $[p_{XY}(x,y)\delta_{x,y}]$ of dimensions $|\mathcal{X}|\times|\mathcal{Y}|$ whose rows  and columns are indexed by the realizations of $X$ and $Y$, respectively, and whose $(x,y)$ entry is $p_{XY}(x,y)\delta_{x,y}$ for all $(x,y)\in\mathcal{X}\times\mathcal{Y}$. We have 
\begin{align}
	\epsilon_1=\max \min_{x,y}\delta_{x,y} \label{theoremeq:maxmin}
\end{align}
where the maximum is over all $\delta_{x,y}\in [0,1]$ such that the matrix $[p_{XY}(x,y)\delta_{x,y}]$ has rank one. 
\end{theorem}
The proof of Theorem \ref{lemma-generaldis1} is given in Section~\ref{eq:lemma-generaldis1proof}.

It remains to show that $S(X;Y\|Z)>0$ if $\epsilon > \epsilon_2$.
Suppose that $\epsilon_2$ is obtained with the minimizer path $(x_1, y_1, x_2, y_2)$ so that
\begin{align}\epsilon_2=\left(\frac{p_{XY}(x_1,y_1)p_{XY}(x_2,y_2)}{p_{XY}(x_1, y_2)p_{XY}(x_2, y_1)}\right)^{1/2}.\label{eqn:92384y}\end{align}
We prove that $S(X;Y\|Z)>0$ for $\epsilon > \epsilon_2$. Since the value of the path $(x_1, y_1, x_2, y_2)$ is less than or equal to the value of the path $(x_1, y_2, x_2, y_1)$, we have
\begin{align}
  \epsilon_2=\left(\frac{
  \min\{p_{11}p_{22}, p_{12}p_{21}\}}{\max\{p_{11}p_{22}, p_{12}p_{21}\}}\right)^{1/2}
\end{align}
where $p_{ij} = p_{XY}(x_i, y_{j})$ for $i,j=1,2$. From Theorem \ref{gentheoremeq:eps12}, we have $S(X;Y\|Z)>0$ if 
\begin{align}
  &\frac 12 D_{\frac 12}\Big(p_{Z|XY}(\cdot |x_1,y_1)\big\| p_{Z|XY}(\cdot |x_{2},y_{2})\Big) <-\log\epsilon_2.
\end{align}
But observe that
\begin{align} 
&\frac12D_{\frac 12}\Big(p_{Z|XY}(\cdot |x_1,y_1)\big\| p_{Z|XY}(\cdot |x_2,y_2)\Big)=-\log\Big(\sum_{z} p_{Z|XY}(z|x_1,y_1)^{\frac 12}\times
p_{Z|XY}(z|x_2,y_2)^{\frac 12}\Big)
\nonumber\\&\quad\leq-\log\Big(p_{Z|XY}(\mathtt{e}|x_1,y_1)^{\frac 12}\times
p_{Z|XY}(\mathtt{e}|x_2,y_2)^{\frac 12}\Big)=-\log(\epsilon)\label{eqnN149}
\end{align}
which proves that $S(X;Y\|Z)>0$ if $\epsilon > \epsilon_2$.


\subsection{Proof of Theorem~\ref{lemma-generaldis1}} \label{eq:lemma-generaldis1proof}
Let 
\begin{align}
	\tilde{\epsilon}_1=\max \min_{x,y}\delta_{x,y}\label{eq:epsilon3}
\end{align}
where the maximization is over all $\delta_{x,y}\in [0,1]$ such that the matrix $[p_{XY}(x,y)\delta_{x,y}]$ has rank one. We prove that $\tilde{\epsilon}_1=\epsilon_1$. 

Observe that if $p_{XY}(x^*,y^*)=0$ for some $x^*,y^*\in\mathcal{X}\times\mathcal{Y}$, then $\epsilon_1=0$, which follows from Definition~\ref{def:path}. We now prove that $\tilde{\epsilon}_1$ is also zero. Consider some arbitrary $\delta_{x,y}$ such that $[p_{XY}(x,y)\delta_{x,y}]$ has rank one. Since $p_{XY}(x^*,y^*)\delta_{x^*,y^*}=0$, we must have either $p_{XY}(x^*,y)\delta_{x^*,y}=0$ for all $y\in\mathcal{Y}$ or $p_{XY}(x,y^*)\delta_{x,y^*}=0$ for all $x\in\mathcal{X}$. Assume that $p_{XY}(x^*,y)\delta_{x^*,y}=0$ for all $y\in\mathcal{Y}$. Since there exists a $y$ such that $p_{XY}(x^*,y)>0$, we obtain  $\delta_{x^*,y}=0$ for some $y\in\mathcal{Y}$. Hence, $\min_{x,y}\delta_{x,y}=0$ and $\tilde{\epsilon}_1=0$. 

Based on the discussions above, we may assume that $p_{XY}(x,y)>0$ for all $(x,y)\in\mathcal{X}\times\mathcal{Y}$. In this case, it follows that $\epsilon_1>0$. We also have $\tilde{\epsilon}_1>0$ since one valid choice for $\delta_{x,y}$ is $\delta_{x,y}=\frac{k}{p_{XY}(x,y)}$, where $k= \min_{x,y}p_{XY}(x,y)$. Since $\tilde{\epsilon}_1>0$, we take the maximum in (\ref{eq:epsilon3}) only over positive $\delta_{x,y}$.

Consider some $\delta_{x,y}\in (0,1]$ such that $[p_{XY}(x,y)\delta_{x,y}]$ has rank one. In other words, $p_{XY}(x,y)\delta_{x,y}=e^{m(x)}e^{n(y)}$ has a product form for some $m(x)$ and $n(y)$. Taking logarithms, we obtain 
\begin{align}\log p_{XY}(x,y)+\log\delta_{x,y}=m(x)+n(y).\end{align} We can express the problem as finding the maximum value of $\tilde{\epsilon}_1\in(0,1]$ such that for some $m(x)$ and $n(y)$, we have for $\forall x,y$
\begin{align}
m(x)\!+\!n(y)\!\leq\! \log p_{XY}(x,y)\!\leq\! m(x)\!+\!n(y)\!-\!\log\tilde{\epsilon}_1.
\end{align}
We can view this as a linear programming problem to minimize $A\triangleq-\log\tilde{\epsilon}_1$ subject to
\begin{align}
	&m(x)+n(y)\leq \log p_{XY}(x,y)\label{eq:topequation}\\
	&\log p_{XY}(x,y)\leq m(x)+n(y)+A.\label{eq:bottomequation}
\end{align}
We consider the dual of this linear programming problem. Multiplying (\ref{eq:topequation}) by some $\gamma(x,y)\geq 0$ and (\ref{eq:bottomequation}) by some $\mu(x,y)\geq 0$, we obtain
\begin{align}
	&\sum_{x,y}\gamma(x,y)\left( m(x)+n(y)\right) + \sum_{x,y}\mu(x,y)\log p_{XY}(x,y)\nonumber\\
	&\quad\leq \sum_{x,y}\gamma(x,y)\log p_{XY}(x,y)+\sum_{x,y}\mu(x,y)\Big( m(x)+n(y)+A\Big).
\end{align}
Since we are interested in the best lower bound on $A$, we should choose $\gamma$ and $\mu$ such that the  coefficient of $A$ is equal to one and the coefficients of free variables $m(x)$ and $n(y)$ vanish. The coefficient of $A$ is one only if  $\sum_{x,y}\mu(x,y)=1$. Furthermore,  to cancel out the auxiliary variables $m(x)$ and $n(y)$ from both sides, we must have $\sum_x \gamma(x,y)=\sum_x \mu(x,y)$ for all $y\in\mathcal{Y},$ and $\sum_y \gamma(x,y)=\sum_y \mu(x,y)$ for all $x\in\mathcal{X}$. This implies that $\gamma(x,y)$ and $\mu(x,y)$ are \emph{probability distributions} with the same marginals. We denote their marginal probabilities by $\mu(x)=\gamma(x) $ and $\mu(y)=\gamma(y)$. 

The dual of the linear programming problem is 
\begin{align}
A\!=\!\max \!\sum_{x,y}\!(\mu(x,y)\!-\!\gamma(x,y))\log p_{XY}(x,y)\label{eqn:34re}
\end{align}
where the maximization is over all pmfs $\mu, \gamma$ with the same marginals. Since  $A=-\log\tilde{\epsilon}_1$, we have
\begin{align}
\tilde{\epsilon}_1=\min \prod_{x,y}p_{XY}(x,y)^{-\mu(x,y)+\gamma(x,y)}\label{eqn:34wewere}
\end{align}
where the minimization is over all pmfs $\mu, \gamma$ with the same marginals. 

The requirement that the pmfs $\mu, \gamma$ should have the same marginals imposes a number of linear constraints on $\mu$ and $\gamma$. This indicates that the set of all pmfs $\mu, \gamma$ with the same marginals is a polytope. We maximize a linear equation over this polytope in \eqref{eqn:34re}.

We first list three claims. These are proved below and used to show that $\tilde{\epsilon}_1=\epsilon_1$.

\textbf{Claim 1:} \emph{If $(\mu, \gamma)$ is a minimizer for (\ref{eqn:34wewere}) and 
\begin{align}\mu(x,y)&=\lambda \mu_1(x,y)+(1-\lambda) \mu_2(x,y)
\\\gamma(x,y)&=\lambda \gamma_1(x,y)+(1-\lambda) \gamma_2(x,y)\end{align}
where $\mu_i, \gamma_i$ are pmfs with the same marginals for $i=1,2$ and $0<\lambda<1$, then $\mu_i, \gamma_i$ are also minimizers for $i=1,2$.}

Given any pmf $\mu$, we define $\textit{Support}(\mu)$ as the set of realizations with positive occurrence probability.

\textbf{Claim 2:} \emph{Given pmfs $(\mu, \gamma)$ with the same marginals, if one can find pmfs $\mu_1(x,y), \gamma_1(x,y)$ with the same marginals such that $\text{Support}(\mu_1)\subseteq \text{Support}(\mu)$ and
$\text{Support}(\gamma_1)\subseteq \text{Support}(\gamma)$, then there is a $\lambda$ with $0<\lambda<1$ and a $(\mu_2, \gamma_2)$ with the same marginals such that \begin{align}\mu(x,y)&=\lambda \mu_1(x,y)+(1-\lambda) \mu_2(x,y)\\ \gamma(x,y)&=\lambda \gamma_1(x,y)+(1-\lambda) \gamma_2(x,y).\end{align}}
\textbf{Claim 3:} \emph{Given pmfs $\mu, \gamma$ with the same marginals, one can find a path  $(x_1, y_1, x_2, y_2, \cdots, x_k, y_k)$ (as in Definition \ref{def:path}) such that $\gamma(x_i, y_i)>0$ for $1\leq i\leq k$, and $\mu(x_1, y_k)>0$ and $\mu(x_i,y_{i-1})>0$ for $2\leq i\leq k$. }

\vspace{0.5cm}

Consider a minimizer $(\mu, \gamma)$ and the path given in Claim 3 for $(\mu, \gamma)$. Define $\mu_1$ and $\gamma_1$ as follows: \begin{align}
\gamma_1(x,y)=\begin{cases}\frac{1}{k}&  \text{if } (x,y)=(x_i, y_i) \text{ for } 1\leq i\leq k,\\
0&  \text{otherwise},
\end{cases}\label{eq:claim31}
\end{align}
\begin{align}
\mu_1(x,y)=\begin{cases}\frac{1}{k}&  \text{if } (x,y)\!=\!(x_1, y_k),\\
\frac{1}{k}&  \text{if } (x,y)\!=\!(x_i, y_{i-1})  \text{ for } 2\!\leq\! i\!\leq\! k,\\
0&  \text{otherwise.} 
\end{cases}\label{eq:claim32}
\end{align}

Observe that $\gamma_1$ and $\mu_1$ have the same marginals and satisfy the conditions $\text{Support}(\mu_1)\subseteq \text{Support}(\mu)$ and $\text{Support}(\gamma_1)\subseteq \text{Support}(\gamma)$. Using Claim 1 and 2, we conclude that $(\gamma_1, \mu_1)$ must also be a minimizer. Using \eqref{eqn:34wewere}, we therefore obtain 
\begin{align}
\tilde{\epsilon}_1&=\prod_{x,y}p_{XY}(x,y)^{-\mu_1(x,y)+\gamma_1(x,y)}
\end{align}
which evaluates to the value assigned to the path $(x_1, y_1, x_2, y_2, \cdots, x_k, y_k)$ according to \eqref{eq:assignedvalue}. Since $\epsilon_1$ is the minimum assigned value of all possible paths, we obtain $\epsilon_1\leq \tilde{\epsilon}_1$. 

To show that $\tilde{\epsilon}_1\leq \epsilon_1$, suppose that $\epsilon_1$ is obtained for the minimizer path $(x'_1, y'_1, x'_2, y'_2, \cdots, x'_k, y'_k)$. We  construct $\mu'_1$ and $\gamma'_1$ for this path, similar to \eqref{eq:claim31} and \eqref{eq:claim32}, and the value of this path is
\begin{align}
\epsilon_1=\prod_{x,y}p_{XY}(x,y)^{-\mu'_1(x,y)+\gamma'_1(x,y)} .
\end{align}
Using \eqref{eqn:34wewere}, we obtain $\tilde{\epsilon}_1\leq \epsilon_1$. This proves $\tilde{\epsilon}_1=\epsilon_1$. 

It remains to prove the claims given above.

\emph{Proof of Claim 1:} The value of $\sum_{x,y}(\mu_i(x,y)-\gamma_i(x,y))\log p_{XY}(x,y)$ must be less than or equal to $\sum_{x,y}(\mu(x,y)-\gamma(x,y))\log p_{XY}(x,y)$ for $i=1,2$ since $(\mu,\gamma)$ is a maximizer for (\ref{eqn:34re}). On the other hand, by the linearity of \eqref{eqn:34re}, we have
\begin{align}
&\sum_{x,y}(\mu(x,y)-\gamma(x,y))\log p_{XY}(x,y)\nonumber\\
&\quad=\lambda\sum_{x,y}(\mu_1(x,y)-\gamma_1(x,y))\log p_{XY}(x,y)\!+\!(1\!-\!\lambda)\sum_{x,y}(\mu_2(x,y)\!-\!\gamma_2(x,y))\log p_{XY}(x,y).
\end{align}
We thus have for $i=1,2$ that
\begin{align}
	&\sum_{x,y}(\mu_i(x,y)-\gamma_i(x,y))\log p_{XY}(x,y)=\sum_{x,y}(\mu(x,y)-\gamma(x,y))\log p_{XY}(x,y).
\end{align}

\emph{Proof of Claim 2:} Assign 
\begin{align}
\lambda\!=\!\min\left\{1, \min_{\substack{x,y:\\ \gamma_1(x,y)>0}}\frac{\gamma(x,y)}{\gamma_1(x,y)}, \min_{\substack{x,y:\\ \mu_1(x,y)>0}}
\frac{\mu(x,y)}{\mu_1(x,y)}\right\}.
\end{align}
Observe that $\lambda>0$ since $\gamma_1(x,y)>0$ implies that $\gamma(x,y)>0$, and $\mu_1(x,y)>0$ implies that $\mu(x,y)>0$. Assigning the values
\begin{align}
\gamma_2(x,y)&=\frac{\gamma(x,y) -\lambda \gamma_1(x,y)}{1-\lambda}\\
\mu_2(x,y)&=\frac{\mu(x,y) -\lambda \mu_1(x,y)}{1-\lambda}
\end{align}  proves the claim.

\emph{Proof of Claim 3:} Consider some $x_1\in\mathcal{X}$ such that $\gamma(x_1)>0$. Then there is some $y_1\in\mathcal{Y}$ such that $\gamma(x_1, y_1)>0$. Hence $\gamma(y_1)>0$, which implies that $\mu(y_1)>0$. This also implies that there is some $x_2\in\mathcal{X}$ such that $\mu(x_2, y_1)>0$. Hence, $\mu(x_2)>0$ implies $\gamma(x_2)>0$ and that there is some $y_2$ such that $\gamma(x_2, y_2)>0$. We continue this process and obtain a sequence $(x_1, y_1, x_2, y_2, ....)$. While applying the process, we must observe at some point for the first time a previously occurred symbol. Suppose that this happens at time $m$. If $x_m=x_i$ for some $i<m$, then we consider the sequence $(x_i, y_i, x_{i+1}, y_{i+1}, ..., x_{m-1}, y_{m-1})$ as our path. This is a desirable path since $\mu(x_i, y_{m-1})=\mu(x_m, y_{m-1})>0$. Similarly, if $y_m=y_i$ for some $i<m$, we consider the sequence $(x_{i+1}, y_{i+1}, ..., x_{m}, y_{m})$ as our path. This is a desirable path since $\mu(x_{i+1}, y_{m})=\mu(x_{i+1}, y_{i})>0$. This proves the existence of such a path.

\subsection{Proof of Theorem~\ref{theorem-lowerbounds}}\label{subsec:proofdsbe22}
\textbf{Part 1:} We first show that the  one-way SK rate from Alice to Bob  is positive if and only if
\begin{align}
	\epsilon> 1-\eta(p_{Y|X})
\end{align}	
where $\eta(\cdot)$ is as defined in (\ref{eqnAG687rr657}). 

The one-way SK rate is positive if and only if one can find auxiliary random variables $U$ and $V$ that satisfy the Markov chain $UV\rightarrow X\rightarrow YZ$ such that 
\begin{align}
	&I(U;Y|V)>I(U;Z|V)=\!(1\!-\!\epsilon)I(U;XY|V)\!=(1-\epsilon)I(U;X|V).
\end{align}
Thus, the one-way SK rate is positive if and only if
\begin{align}
	&\epsilon>1-\sup_{UV\rightarrow X\rightarrow Y}\frac{I(U;Y|V)}{I(U;X|V)}\overset{(a)}{=} 1-\eta(p_{Y|X}) 
\end{align}
where $(a)$ follows by Lemma \ref{lemma2fdret}, proved in Appendix \ref{appendix-lemmas}.

\noindent
\textbf{Part 2:} We next prove that $\bar{L}(X;Y\|Z)=0$  if and only if
\begin{align}
	\epsilon\leq 1-\max_{\substack{q_{XY}:\: q_{XY}\preceq p_{XY}}}\rho^2_m(q_{XY}).
\end{align}

From the definition of $\bar{L}(X;Y\|Z)=0$, one can deduce that $\bar{L}(X;Y\|Z)=0$ if and only if, for any $U_1, U_2, \cdots, U_k$ satisfying \eqref{eq:LowerboundoddiMarkov} and \eqref{eq:LowerboundeveniMarkov}, any $i$ with $1\leq i\leq k$, and any $u_{1:i-1}$ such that $\mathbb{P}[U_{1:i-1}=u_{1:i-1}]>0$, we have for odd and even $i$, respectively,
\begin{align}
& I(U_i; Y|U_{1:i-1}\!=\!u_{1:i-1}) \le \!I(U_i;Z|U_{1:i-1}\!=\!u_{1:i-1}) \label{eq231}\\
 & I(U_i; X|U_{1:i-1}\!=\!u_{1:i-1})\le \!I(U_i;Z|U_{1:i-1}\!=\!u_{1:i-1})\label{eq232}.
 \end{align}
The reason is that if either \eqref{eq231} or \eqref{eq232} fails, we can construct a valid $U_i$
by setting $U_i$ to a constant if $U_{1:i-1}\neq u_{1:i-1}$ so that 
\begin{align}
    &I(U_i;Y|U_{1:i-1})=\mathbb{P}[U_{1:i-1}=u_{1:i-1}] I(U_i; Y|U_{1:i-1}\!=\!u_{1:i-1}).
\end{align}
We can then compute a non-zero lower bound ${L}(X;Y\|Z)$ by  considering $k=\zeta=i$ in 
\eqref{eq:lowerbound}.

Equivalently, one can show that $\bar{L}(X;Y\|Z)=0$ if and only if for any $U_1, U_2, \cdots, U_k$ satisfying \eqref{eq:LowerboundoddiMarkov} and \eqref{eq:LowerboundeveniMarkov}, any $i$ with $1\leq i\leq k$, and any $u_{1:i-1}$ such that $\mathbb{P}[U_{1:i-1}=u_{1:i-1}]>0$, the one-way SK rates for the distribution
\begin{align}
  r_{XY}(\cdot)=p_{XY|U_{1:i-1}}(\cdot | u_{1:i-1})
\end{align}
are zero. This is because for any $r_{UVXY}=r_{UV|X}\, r_{Y|X}$ such that $I_r(U;Y|V)-I_r(U;Z|V)>0$, there exists some $v$ such that $I_r(U;Y|V=v)-I_r(U;Z|V=v)>0$. 
If $i$ is odd, we can then append to $U_{1:i-1}$ the choices $U_i=V$ and $U_{i+1}=U$. Considering $U_{1:i}=(u_{1:i-1}, v)$, we have
\begin{align}
   I(U_{i+1}; Y|U_{1:i}\!=\!u_{1:i})\!-\!I(U_{i+1};Z|U_{1:i}\!=\!u_{1:i})> 0.
\end{align}
For even $i$, we can set $U_i=\text{constant}$, $U_{i+1}=V$, and $U_{i+2}=U$, and proceed similarly.

To complete the proof, we need to characterize the class of pmfs $r_{XY}$ that arises when we condition the joint pmf of $(X,Y)$ on  $U_{1:i-1}=u_{1:i-1}$. The authors of \cite{MaIshwar2013,Verdu} consider this problem, where they search for the set of  conditional pmfs  $p_{XY|U_{1:k}}$ that one can obtain with some auxiliary random variables $U_1, U_2, \cdots, U_k$ satisfying
\begin{align*}
	&U_i\rightarrow XU_{1:i-1}\rightarrow Y \quad \text{for odd $i$},\nonumber\\
	&U_i\rightarrow YU_{1:i-1}\rightarrow X \quad \text{for even $i$}
\end{align*}
for some arbitrary $k$ and arbitrary realization $u_{1:k}$ of $U_{1:k}$ satisfying $\mathbb{P}[U_{1:k}=u_{1:k}]>0$. This set of pmfs can be expressed as $q_{XY}(x,y)=a(x)b(y)p_{XY}(x,y)$ for some functions $a:\mathcal{X}\rightarrow \mathbb{R}$ and $b:\mathcal{Y}\rightarrow \mathbb{R}$ \cite{MaIshwar2013, Verdu}, i.e., $q_{XY}\preceq p_{XY}$. Combining this observation with (\ref{eq:Theorem6part1}) proves that $\bar{L}(X;Y\|Z)=0$  if and only if
\begin{align}
  \epsilon & \le 1-\max_{\substack{q_{XY}:\: q_{XY}\preceq p_{XY}}}\eta(q_{Y|X}) \overset{(a)}{=}1-\max_{\substack{q_{XY}:\: q_{XY}\preceq p_{XY}}}\rho^2_m(q_{XY})\end{align}
where $(a)$ follows by \eqref{eqnAG687rr657} and because if $q_{XY}\preceq p_{XY}$, then for any 
$r_{XY}=r_{X}\, q_{Y|X}$ we also have $r_{XY}\preceq p_{XY}$. A similar argument is used in \cite[Eq. (80)]{Verdu}.

\noindent
\textbf{Part 3:} Since $p_{Z|XY}$ is an erasure channel with probability $\epsilon$, Lemma \ref{lemmmasthm1}  in Appendix \ref{appendix-lemmas} shows that  a given conditional pmf $p_{\bar{Z}|X,Y}$ can be produced with a degradation $p_{\bar{Z}|Z}$ on random variable $Z$, if and only if 
\begin{align}
\epsilon\leq \sum_{\bar{z}} \min_{x,y:\: p_{XY}(x,y)>0} p_{\bar{Z}|X,Y}(\bar{z}|x,y).
\end{align}
As a result, the intrinsic mutual information $B_0(X;Y\|Z)$ in (\ref{eq:B0upperbound}) is zero if and only if 
\begin{align}
\epsilon\leq \sup_{\substack{p_{\bar{Z}|X,Y}:\\ I(X;Y|\bar{Z})=0}}\sum_{\bar{z}} \min_{x,y:\: p_{X,Y}(x,y)>0} p_{\bar{Z}|X,Y}(\bar{z}|x,y).
\end{align}
In computing $B_0(X;Y\|Z)$, it suffices to restrict to random variables $\bar{Z}$ with cardinality at most $|\mathcal{Z}|$ \cite{christandl2003property}. Therefore, we assume that $\bar{z}\in\{1,2,\cdots, |\mathcal{Z}|\}$.
Finally, observe that the condition $I(X;Y|\bar{Z})=0$ is equivalent to the condition that the matrix 
$$[p_{XY}(x,y)p_{\bar{Z}|X,Y}(\bar{z}|x,y)]_{x,y}$$
has rank one  for all $\bar{z}$.

\noindent
\textbf{Part 4:} Consider an arbitrary $p_{J|XY}$. The bound $ B_1(X;Y\|Z)$ in (\ref{eq:B1upperbound}) is zero only if $I(X;Y|J)=0$. Thus, assume that for $p_{X,Y,J}=p_{X,Y}p_{J|X,Y}$, we have $X\rightarrow J\rightarrow Y$ forming a Markov chain. Since $Z$ is the result of $XY$ passing through an erasure channel, for any $p_{UVXY}\, p_{Z|XY}$ we have
\begin{align}
  I(U;Z|V)=(1-\epsilon)I(U;XY|V).
\end{align}
Thus, $ B_1(X;Y\|Z)$ is zero if and only if for any $p_{UVXY}p_{J|X,Y}$ we have
\begin{align}(1-\epsilon)I(U;XY|V)\geq I(U;J|V).\label{eqnAAAcv1}
\end{align}
We claim that this is equivalent with the following condition:  for any $q_{U,X,Y}p_{J|X,Y}$ we must have
\begin{align}(1-\epsilon)I(U;XY)\geq I(U;J).\label{eqnAAAcv2}
\end{align}
Clearly \eqref{eqnAAAcv2} implies \eqref{eqnAAAcv1} for \eqref{eqnAAAcv2} implies that 
\begin{align}
  (1-\epsilon)I(U;XY|V=v)\geq I(U;J|V=v), \quad \forall v.
\end{align}
For the other direction, assume that \eqref{eqnAAAcv1} holds for  any $p_{UVXY}p_{J|X,Y}$. Take some $v\in\mathcal{V}$ and assume that $U$ is a constant when $V\neq v$, while we let $p(u|V=v)$ to be arbitrary. Then, \eqref{eqnAAAcv1} implies that 
\begin{align}(1-\epsilon)I(U;XY|V=v)\geq I(U;J|V=v).\label{eqnAAAcv1p}
\end{align}
Since $p(x,y)>0$, for any $q_{XY}$ one can find $p_{VXY}$ such that $p(x,y|V=v)=q_{XY}$. Since 
$p(u|V=v)$ was arbitrary, we obtain \eqref{eqnAAAcv2}. 

The condition \eqref{eqnAAAcv2} can be expressed in terms the strong data processing constant and Renyi's maximal correlation as
\begin{align}
	1-\epsilon \geq \max_{q_{XY}\, p_{J|XY}}s^*(q_{XYJ})=  \eta(p_{J|X,Y}).
\end{align}

\subsection{Proof of Theorem~\ref{thm:A5}}\label{subsec:proofdsbe}
Using Theorem~\ref{theorem-lowerbounds},  the one-way SK rate for a DSBE source is positive if
\begin{align}
  \epsilon>1-\max_{\substack{q_{XY}:\: q_{XY}=q_{X}p_{Y|X}}}\rho^2_m(q_{XY}).
\end{align}
In particular, the one-way SK rate for a DSBE source is positive if
\begin{align}
  \epsilon>1-\rho^2_m(p_{XY})=1-{(1-2p)}^2=4p(1-p).
\end{align}
We next prove that $\bar{L}(X;Y\|Z)=0$ if $\epsilon\leq 4p(1-p)$, which completes the proof of Theorem~\ref{thm:A5}. Note that the lower bound $\bar{L}(X;Y\|Z)$ obtained from (\ref{eq:lowerbound}) includes the one-way SK rate. Using \cite[Theorem 6]{Verdu}, we obtain
\begin{align}
\max_{\substack{q_{XY}:\: q_{XY}\preceq p_{XY}}}\rho_m^2(q_{XY})&=(1-2p)^2 \label{eq:maximalcorrforDSBE}.
\end{align}
Combining (\ref{eq:theorem4lowerbound}) with (\ref{eq:maximalcorrforDSBE}) gives  the desired result.

\bibliographystyle{IEEEtran}
\bibliography{IEEEabrv,References}

\appendices
\section{A New Interpretation of the Upper Bound $B_1(X;Y\|Z)$}\label{appendixA}
In this section, we give a new interpretation of the upper bound $B_1(X;Y\|Z)$. In \cite{ourpaper}, the bound $S(X;Y\|Z)\leq B_1(X;Y\|Z)$ follows by showing that
\begin{align}
S(X;Y\|Z)&\leq S(X;Y\|J)+S_{\text{ow}}(XY; J\|Z)\nonumber
\\&= S(X;Y\|J)+\!\max_{UV\rightarrow XY\rightarrow ZJ}I(U;J|V)\!-\!I(U;Z|V)\label{less-noisy-djf}
\end{align}
where $S_{\text{ow}}(XY; J\|Z)$ is the one-way SK rate from $XY$ to $J$. The interpretation of \eqref{less-noisy-djf} given in \cite{ourpaper} is to split the key shared between $X$ and $Y$ (and hidden from $Z$) into two parts: a part independent of $J$ (i.e., the term $S(X;Y\|J)$) and another part shared with $J$ (i.e., the term $S_{\text{ow}}(XY; J\|Z)$). 

We now give a new interpretation for (\ref{less-noisy-djf}). To do this, we begin by revisiting the intrinsic mutual information upper bound $B_0(X;Y\|Z)$. For this bound, note that making Eve weaker by passing $Z$ through a channel $p_{J|Z}$ does not decrease the SK capacity. Thus, $S(X;Y\|Z)\leq S(X;Y\|J)$ if $J\rightarrow Z\rightarrow XY$ forms a Markov chain. We now replace the \emph{degradation} of $Z$ to $J$ with the \emph{less noisy} condition. Consider a broadcast channel with input $(X,Y)$ and two outputs $Z$ and $J$. We have the following proposition:
\begin{proposition} 
	\label{prop134}
	If the channel $p_{Z|XY}$ is less noisy than the channel $p_{J|XY}$, then 
	\begin{align}
	S(X;Y\|Z)\leq S(X;Y\|J).\label{prop:lessnoisy}
	\end{align}
\end{proposition}

\begin{proof}[Proof 1]
	The proposition follows from \eqref{less-noisy-djf} since if $Z$ is less noisy than $J$, then $I(U;J|V)-I(U;Z|V)$ vanishes in (\ref{less-noisy-djf}). We give a direct proof of Proposition \ref{prop134} below. 
\end{proof}

\begin{proof}[Proof 2]
	Suppose that we have a code for $(X^n, Y^n, Z^n)$ with public communication $F_1, F_2, \cdots, F_k$, and keys $K_A$ and $K_B$. We then have
\begin{align}
   \frac 1n I(K_A;Z^nF_{1:k})\leq \epsilon.
\end{align}
	We now show that
	\begin{align}
	I(K_A;J^nF_{1:k})\leq I(K_A;Z^nF_{1:k})\label{eq:claimlessnoisy}
	\end{align}
	which shows that the code is secure also for an eavesdropper that has i.i.d.\ repetitions of $J$ instead of $Z$. To prove (\ref{eq:claimlessnoisy}), we need to show that 
\begin{align}
   I(K_A;J^n|F_{1:k})\leq I(K_A;Z^n|F_{1:k}).
\end{align}
	It therefore suffices to show for all $f_{1:k}$ that
	\begin{align}
	&I(K_A;J^n|F_{1:k}=f_{1:k})\leq I(K_A;Z^n|F_{1:k}=f_{1:k}).\label{eqn:rev1}\end{align}
	Since we have the Markov chain $F_{1:k}K_A\rightarrow X^nY^n\rightarrow Z^nJ^n$, when we condition on $F_{1:k}=f_{1:k}$, we also have the Markov chains
	\begin{align}
	    \begin{array}{l}
	      K_A \rightarrow X^nY^n\rightarrow Z^nJ^n \\
	      F_{1:k} \rightarrow X^nY^n\rightarrow Z^nJ^n.
	    \end{array}
	\end{align}
Since $p_{Z|XY}$ is less noisy than $p_{J|XY}$, the $n$-letter product channel $p_{Z^n|X^nY^n}$ is also less noisy than the channel $p_{J^n|X^nY^n}$. Thus, we have the bound \eqref{eqn:rev1}.
\end{proof}

Proposition \ref{prop134} gives the following interpretation of \eqref{less-noisy-djf}: the term $S_{\text{ow}}(XY; J\|Z)$ is the penalty of deviating from the less-noisy condition when we replace $Z$ with $J$ in $S(X;Y\|Z)$ and $S(X;Y\|J)$.

\section{A New Measure of Correlation}\label{appendixB}
The quantity $\epsilon_2$ given in Definition \ref{def:path} motivates a new measure of correlation. 
Suppose we are given a joint pmf $p_{XY}$. 
Let
\begin{align*}
   q_{X_1Y_1X_2Y_2}(x_1,y_1,x_2,y_2) & =p_{XY}(x_1,y_1)p_{XY}(x_2,y_2) \\
   r_{X_1Y_1X_2Y_2}(x_1,y_1,x_2,y_2) & =p_{XY}(x_1, y_2)p_{XY}(x_2, y_1)
\end{align*}
and define the new correlation measure
\begin{align}
	&J_\alpha(X;Y)\triangleq D_{\alpha}\Big({q_{X_1Y_1X_2Y_2}}\big\|{r_{X_1Y_1X_2Y_2}}\Big)\label{eq:Japlhadef}
\end{align}
where $D_\alpha$ is the R{\'e}nyi divergence of order $\alpha$. We have
\begin{align}
&\exp(J_{\infty}(X;Y))=\max_{\substack{x_1, x_2, y_1, y_2:\\
\substack{p_X(x_1), p_X(x_2)>0,\\ p_Y(y_1), p_Y(y_2)>0}}
}\left(\frac{p_{XY}(x_1,y_1)p_{XY}(x_2,y_2)}{p_{XY}(x_1, y_2) p_{XY}(x_2, y_1)}\right) \label{eqn:J134}
\nonumber\\
&\quad= \max_{\substack{x_1\neq x_2, y_1\neq y_2:\\
\substack{p_X(x_1), p_X(x_2)>0,\\ p_Y(y_1), p_Y(y_2)>0}}}
\left(\frac{p_{XY}(x_1,y_1)p_{XY}(x_2,y_2)}{p_{XY}(x_1, y_2) p_{XY}(x_2, y_1)}\right)=\frac{1}{\epsilon_2^2}.
\end{align}
Some properties of $J_\alpha$ are as follows.
\begin{itemize}
\item \emph{Faithfulness}: $J_{\alpha}(X;Y)\geq 0$ with equality $J_{\alpha}(X;Y)=0$ if and only if $X$ and $Y$ are independent. Equality  follows from
\begin{align}
   & q_{X_1Y_1X_2Y_2}(x_1,y_1,x_2,y_2)  = r_{X_1Y_1X_2Y_2}(x_1,y_1,x_2,y_2)
\end{align}
for all $x_1,x_2,y_1,y_2$.
{Data Processing}: If $A\rightarrow X\rightarrow Y\rightarrow B$, then we have $J_{\alpha}(X;Y)\geq J_{\alpha}(A;B)$. In particular, we have
\begin{align}
  J_{\alpha}(X_1X_2; Y_1Y_2)\geq J_{\alpha}(X_1;Y_1).
\end{align}
To show that $A\rightarrow X\rightarrow Y\rightarrow B$ implies $J_{\alpha}(X;Y)\geq J_{\alpha}(A;B)$,
consider some $p_{A|X}$ and $p_{B|Y}$. Let
\begin{align}
    &V(a_1a_2b_1b_2|x_1x_2y_1y_2) = p_{A|X}(a_1|x_1)  p_{A|X}(a_2|x_2) p_{B|Y}(b_1|y_1) p_{B|Y}(b_2|y_2)
\end{align}
and define
\begin{align}
&q_{A_1B_1A_2B_2}(a_1,b_1,a_2,b_2)=\sum_{x_1,x_2,y_1,y_2}\Big(q_{X_1Y_1X_2Y_2}(x_1,y_1,x_2,y_2) V(a_1a_2b_1b_2|x_1x_2y_1y_2)\Big)
\nonumber\\&\quad =p_{AB}(a_1,b_1)p_{AB}(a_2,b_2)
\end{align}
and
\begin{align}
&r_{A_1B_1A_2B_2}(a_1,b_1,a_2,b_2)=\sum_{x_1,x_2,y_1,y_2}\Big(r_{X_1Y_1X_2Y_2}(x_1,y_1,x_2,y_2) V(a_1a_2b_1b_2|x_1x_2y_1y_2)\Big)
\nonumber\\&\quad=p_{AB}(a_1,b_2)p_{AB}(a_2,b_1).
\end{align}
The bound $J_{\alpha}(X;Y)\geq J_{\alpha}(A;B)$ follows from the data processing property of $D_\alpha$.

\item \emph{Symmetry}: The definition (\ref{eq:Japlhadef}) implies that 
\begin{align}
  J_{\alpha}(X;Y)=J_{\alpha}(Y;X).
\end{align}
\item \emph{Additivity}: If $(X_1, Y_1)$ and $(X_2, Y_2)$ are independent, then
\begin{align}
  J_{\alpha}(X_1X_2; Y_1Y_2)=J_{\alpha}(X_1;Y_1)+J_{\alpha}(X_2;Y_2)
\end{align}
which follows from (\ref{eq:Japlhadef}).
\end{itemize}

Furthermore, $J_\infty(X;Y)$ has the following properties:
\begin{itemize}
\item By \eqref{eqn:J134} we see that $J_\infty(X;Y)$ depends only on $p_{Y|X}$ and the support set of the distribution on $X$. Thus, with abuse of notation, we may write
\begin{align}
J_\infty(p_{Y|X}(y|x), \mathcal{X}')\triangleq J_\infty(X;Y)
\end{align}
where
$\mathcal{X}'=\{x:p_X(x)>0\}$.
\item For a pmf $p_{XYZ}$, let 
\begin{align}
    J_{\infty}(X;Y|Z)\triangleq \max_{z: p_Z(z)>0} J_\infty(X;Y|Z=z).
\end{align}
Then, if $F-X-Y$ forms a Markov chain, we have
\begin{align}
  J_{\infty}(X;Y)\geq J_\infty(X;Y|F).
\end{align}
\end{itemize}
An application of $J_{\infty}$ is given in the next subsection. As another application, 
consider a key agreement protocol of blocklength $n$ with public messages $F_1, F_2, \ldots, F_k$ to agree on single bits $K_A, K_B\in\{0,1\}$. Then we can write
\begin{align}n J_{\infty}(X;Y|Z)&=J_{\infty}(X^n;Y^n|Z^n) \nonumber
\\&\geq J_{\infty}(X^n;Y^n|F_1Z^n) \nonumber
\\&\geq J_{\infty}(X^n;Y^n|F_1F_2Z^n) \nonumber
\\&\quad\ldots \nonumber
\\&\geq J_{\infty}(X^n;Y^n|F_{1:k}Z^n) \nonumber
\\&\geq J_{\infty}(K_A;K_B|F_{1:k}Z^n) \nonumber
\\&\geq \max_{f_{1:k},z^n} \log\Bigg(\frac{p_{K_AK_B|f_{1:k},z^n}(0,0)}{p_{K_AK_B|f_{1:k},z^n}(0,1)}\frac{p_{K_AK_B|f_{1:k},z^n}(1,1)}{p_{K_AK_B|f_{1:k},z^n}(1,0)}\Bigg) 
\end{align}
providing a bound on how fast Alice and Bob can approach the ideal distribution 
\begin{align}
   q_{K_AK_B Z^n F_{1:k}}=\frac{1}{2}\mathds{1}[k_A=k_B]\,p_{Z^n F_{1:k}}
\end{align}
as given in \eqref{quality}.

\subsection{An ``Uncertainty Principle" for Channel Coding}
Consider a point to point communication to send a message $M$ over a channel $p_{Y|X}(y|x)$ with $n$ channel uses. Then, the chain $M\rightarrow X^n\rightarrow Y^n\rightarrow \hat{M}$ is Markov so that the data-processing property of the new correlation measure implies
\begin{align}
   J_{\infty}(M;\hat{M})\leq J_{\infty}(X^n;Y^n).
\end{align}
However, $J_{\infty}(X^n;Y^n)$ depends only on the channel $p(y^n|x^n)$ and the support set of $X^n$, that is $\{x^n:p(x^n)>0\}$. Because increasing the support set can only increase $J_{\infty}$, we conclude that 
\begin{align}
  J_{\infty}(X^n;Y^n) & \le J_{\infty}(p(y^n|x^n), \mathcal{X}^n)  = n\times J_{\infty}(p(y|x), \mathcal{X}).
\end{align}
Thus, we have
\begin{align}
  J_{\infty}(M;\hat{M})\leq n\times J_{\infty}(p(y|x), \mathcal{X})
\end{align}
which implies that for any $m_1\neq m_2$ we have
\begin{align}\log &\left(\frac{p_{\hat{M}|M}(m_1|m_1)}{p_{\hat{M}|M}(m_2|m_1)}\times \frac{p_{\hat{M}|M}(m_2|m_2)}{p_{\hat{M}|M}(m_1|m_2)}\right)\leq n\times J_{\infty}(p(y|x), \mathcal{X}).\label{eqwesdn1f}\end{align}
Equivalently, we have the following inequality
\begin{align} &\frac{p_{\hat{M}|M}(m_2|m_1)}{p_{\hat{M}|M}(m_1|m_1)}\times \frac{p_{\hat{M}|M}(m_1|m_2)}{p_{\hat{M}|M}(m_2|m_2)}\geq \exp\left(-n\times J_{\infty}(p(y|x),    \mathcal{X})\right).\label{eqwesdn1f2}\end{align}
Define
\begin{align*}
	\frac{P_{\hat{M}|M}(m_2|m_1)}{P_{\hat{M}|M}(m_1|m_1)}
\end{align*}
as the ``uncertainty of $M=m_1$" against $M=m_2$ if the true value of $M$ is $m_1$. Similarly, define the ``uncertainty of $M=m_2$" against $M=m_1$ if the true value of $M$ is $m_2$ as
\begin{align*}
	\frac{P_{\hat{M}|M}(m_1|m_2)}{P_{\hat{M}|M}(m_2|m_2)}.
\end{align*}
Then \eqref{eqwesdn1f2} gives a lower bound on the product of the uncertainty of $M=m_1$ and the uncertainty of $M=m_2$. \eqref{eqwesdn1f2} states that ``if the uncertainty when the true value of $M$ is $m_1$ is very small, then the uncertainty when the true value of $M$ is $m_2$ cannot be small." This may be considered as an uncertainty principle. 

\section{Lemmas Used in the Proof of Theorem~\ref{theorem-lowerbounds}}\label{appendix-lemmas}
\begin{lemma}\label{lemma2fdret}
For any $p_{X}$ such that $p_X(x)>0$ for all $x$, and any $p_{Y|X}$, we have
\begin{align}
  \eta\big(p_{Y|X}\big) =\sup_{UV\rightarrow X\rightarrow Y}\frac{I(U;Y|V)}{I(U;X|V)}.
\end{align}

\end{lemma}
\begin{proof}
For any $(U,V)$ satisfying the Markov chain $UV\rightarrow X\rightarrow Y$, we have 
\begin{align}
\frac{I(U;Y|V)}{I(U;X|V)} & \leq \max_{v}\frac{I(U;Y|V=v)}{I(U;X|V=v)}\nonumber
\\&\quad\overset{(a)}{\leq} s^*_r(X;Y)\nonumber
\\&\quad\leq \max_{\substack{q_{XY}:\: q_{XY}=q_{X}p_{Y|X}}}s_q^*(X;Y)\nonumber
\\&\quad=\eta\big(p_{Y|X}\big)\label{eqref12356f}
\end{align}
where $(a)$ follows because $s^*_r(X;Y)$ is the strong data processing constant evaluated according to
\begin{align}r_{XY}(\cdot) = p_{XY|V}(\cdot | v) = p_{X|V}(\cdot | v) p_{Y|X} (\cdot | \cdot).
\end{align}
Therefore, we have
\begin{align}
  \eta\big(p_{Y|X}\big) \geq \sup_{UV\rightarrow X\rightarrow Y}\frac{I(U;Y|V)}{I(U;X|V)}.
\end{align}
On the other hand, consider an arbitrary $p_{V|X}$.
Fix some $v^*$ such that $P_V(v^*)>0$ and let $U$ be a constant when $V\neq v^*$. Then we have
\begin{align}
  \frac{I(U;Y|V)}{I(U;X|V)}= \frac{I(U;Y|V=v^*)}{I(U;X|V=v^*)}.
\end{align}
Thus, we obtain 
\begin{align}
  \sup_{UV\rightarrow X\rightarrow Y}\frac{I(U;Y|V)}{I(U;X|V)}\geq s^*(X;Y|V=v^*).
\end{align}
Since $p_X(x)>0$ for all $x\in\mathcal{X}$, for any pmf $q_X$ on $\mathcal{X}$ one can find a channel $p_{V|X}$ and a value $v^*$ such that $p_{XY|V}(\cdot |v^*)=q_X(\cdot) p_{Y|X}(\cdot | \cdot)$. Thus, we obtain
\begin{align}
 \sup_{UV\rightarrow X\rightarrow Y}\frac{I(U;Y|V)}{I(U;X|V)}
 & \ge \max_{q_X}s^*\big(q_X\, p_{Y|X}\big)  = \eta\big(p_{Y|X}\big).
\end{align}
\end{proof}

\begin{lemma}\label{lemmmasthm1}
Let $\mathcal{A}$ and $\mathcal{R}$ be arbitrary discrete sets. Let $p_{A}$ be a pmf on $\mathcal{A}$ such that $p_{A}(a)>0$ for all $a\in\mathcal{A}$. Let $\mathcal{B}=\mathcal{A}\cup\{\mathtt e\}$ and 
assume that $p_{B|A}$ is an erasure channel with probability $\epsilon$ (here $\mathtt e$ is the erasure symbol). Consider an arbitrary $q_{R|A}$. There exists  a conditional distribution $p_{R|B}$ such that
\begin{align}
	\sum_bp_{A,B}(a,b)p_{R|B}(r|b)=p_{A}(a)q_{R|A}(r|a)\label{eqn:sderf3r3}
\end{align}
for all $a,r$ if and only if
\begin{align}\sum_{r} \min_{a} q_{R|A}(r|a)&\geq \epsilon.\label{eqn:reg0}\end{align}
\end{lemma}

\begin{remark}
	The term  $\sum_{r} \min_{a} q_{R|A}(r|a)$ is known as  Doeblin's coefficient of ergodicity of the channel $q_{R|A}$ \cite[Section 5]{Cohen}. One direction of Lemma~\ref{lemmmasthm1} is mentioned in \cite{makur2018comparison} and \cite[Remark 3.2]{Maxim}; the other direction seems to be new. 
\end{remark}

\begin{proof}
Suppose that we have  a conditional distribution $p_{R|B}$ such that \eqref{eqn:sderf3r3} holds. Since $p_{B|A}$ is an erasure channel, given an input $A=a$, $B$ has two possibilities $B\in\{a,\mathtt{e}\}$. We have
\begin{align}
	p_{A}(a)q_{R|A}(r|a)&=\sum_bp_{A,B}(a,b)p_{R|B}(r|b) \!=\!\epsilon\, p_{A}(a)p_{R|B}(r|\mathtt e)\!+\!(1\!-\!\epsilon)p_{A}(a)p_{R|B}(r|a) \nonumber
\\&\geq\! \epsilon\, p_{A}(a)p_{R|B}(r|\mathtt e).
\end{align}
Thus, from $p_{A}(a)>0$ for all $a\in\mathcal{A}$ we have
$
q_{R|A}(r|a)\geq \epsilon p_{R|B}(r|\mathtt e).
$
We obtain
\begin{align}
\min_{a}q_{R|A}(r|a)\geq \epsilon \, p_{R|B}(r|\mathtt e).
\end{align}
Therefore, we observe that
\begin{align}
\sum_{r}\min_{a}q_{R|A}(r|a)\geq \epsilon \sum_{r}p_{R|B}(r|\mathtt e)=\epsilon
\end{align}
which proves the correctness of \eqref{eqn:reg0}.

Conversely, take some arbitrary $q_{R|A}$ satisfying  \eqref{eqn:reg0} and let $q_{AR}=p_A\, q_{R|A}$. Then for all $r$ we have
\begin{align}
  q_R(r)=\sum_{a}p_A(a)q_{R|A}(r|a)>0
\end{align}
and 
\begin{align}
  \sum_{a}p_A(a)\frac{q_{R|A}(r|a)}{q_{R}(r)}=1.
\end{align}
Thus, we conclude that for any $r$ we have
\begin{align}
  \min_{a}\frac{q_{R|A}(r|a)}{q_R(r)}\leq 1.
\end{align}
Furthermore, using \eqref{eqn:reg0}, we have
\begin{align}\sum_{r}q_R(r)\min_{a}\frac{q_{R|A}(r|a)}{q_R(r)}&\geq \epsilon.\label{eqn:reg0-34wsd}\end{align}
Thus, one can find 
$\lambda(r)\in[0,1]$ such that
\begin{align}\lambda(r)&\leq \min_{a}\frac{q_{R|A}(r|a)}{q_R(r)}  \label{eqn:reg1}
\end{align}
and
\begin{align}
\epsilon&=\sum_{r} q_R(r)\lambda(r).\label{eqn:reg2}\end{align}

Define $p_{R|B}$ as follows:
\begin{align}
p_{R|B}(r|\mathtt{e})&=\frac{q_R(r)\lambda(r)}{\epsilon}\label{eqn:reg3}
\\p_{R|B}(r|a)&=\frac{q_{R|A}(r|a)\!-\!q_{R}(r)\lambda(r)}{1\!-\!\epsilon}, \quad\forall a\in\mathcal{A}\label{eqn:reg4}.
\end{align}
The conditional probability in \eqref{eqn:reg3} is well defined by \eqref{eqn:reg2}. The conditional probability in \eqref{eqn:reg4} is non-negative and well-defined by \eqref{eqn:reg1}. Finally, observe that with the choice of $p_{R|B}(r|b)$ in (\ref{eqn:reg3}) and (\ref{eqn:reg4}), the marginal pmf of $A,R$ has
\begin{align}
p_{AR}(a,r)&=\sum_{b}p_{A,B}(a,b)p_{R|B}(r|b)
\nonumber\\&=\epsilon\, p_A(a)p_{R|B}(r|\mathtt{e})+(1-\epsilon)p_A(a)p_{R|B}(r|a)
\nonumber\\&=p_A(a)q_{R}(r)\lambda(r)\!+\!p_A(a)\big(q_{R|A}(r|a)\!-\!q_R(r)\lambda(r)\big)
\nonumber\\&=p_A(a)q_{R|A}(r|a)
\end{align}
which proves the converse.
\end{proof}

\end{document}